\documentclass[twocolumn, floatfix, nofootinbib, aps, pra, 10pt, superscriptaddress]{revtex4-2}

\usepackage[utf8]{inputenc}
\usepackage[english]{babel}
\usepackage{natbib}
\usepackage{graphicx}
\usepackage{amsmath, amssymb, amsthm}
\usepackage{tensor}
\usepackage{cases}
\usepackage{dsfont}
\usepackage[dvipsnames, table]{xcolor}
\usepackage{hyperref}
\hypersetup{
    colorlinks = true,
    allcolors = RoyalPurple
}
\usepackage[percent]{overpic}
\usepackage{svg}
\usepackage{tikz}
\usetikzlibrary{decorations.pathmorphing}
\usepackage[export]{adjustbox}
\usepackage{leftidx}
\usepackage{orcidlink}
\usepackage[capitalise]{cleveref}
\usepackage{enumitem}

\numberwithin{equation}{section}
\DeclareMathOperator{\sgn}{sgn}
\DeclareMathOperator{\Tr}{Tr}
\DeclareMathOperator{\tr}{tr}

\newtheorem{definition}{Definition}[section]

\newtheorem{lemma}[definition]{Lemma}

\newtheorem{prop}[definition]{Proposition}
\newtheorem{corollary}[definition]{Corollary}
\newtheorem{example}[definition]{Example}

\Crefname{definition}{Definition}{Definitions}
\Crefname{theorem}{Theorem}{Theorems}
\Crefname{corollary}{Corollary}{Corollaries}
\Crefname{example}{Example}{Examples}
\Crefname{prop}{Proposition}{Propositions}

\newcommand{\norm}[1]{\left\lVert #1 \right\rVert}
\newcommand{\abs}[1]{\left\lvert #1 \right\rvert}
\newcommand{\ket}[1]{\lvert #1 \rangle}
\newcommand{\bra}[1]{\langle #1 \rvert}
\newcommand{\braket}[1]{\langle #1\rangle}
\newcommand{\kb}[2]{\lvert #1 \rangle\langle #2 \rvert}
\newcommand{\pro}[1]{\ket{#1}\bra{#1}}

\newcommand{\HH}{\mathcal{H}}
\newcommand{\EE}{\mathcal{E}}
\newcommand{\CC}{\mathbb{C}}
\newcommand{\UO}{\operatorname{U}(1)}
\newcommand{\SUT}{\operatorname{SU}(2)}
\newcommand{\ID}{\mathds{1}}

\definecolor{existcolor}{RGB}{200,255,200}
\definecolor{notexistcolor}{RGB}{255,200,200}

\newcommand{\change}[1]{#1}

\newcommand{\threeVAbasisI}{
\begin{tikzpicture}
    \draw (-0.05,0) arc (-180:0:.2);
    \draw (-0.5,-0.8) arc (150:-30:.2);
    \draw (0.5,-1) arc (210:30:.2);
\end{tikzpicture}
}

\newcommand{\threeVAbasisII}{
\begin{tikzpicture}
    \draw (-0.05,0) arc (-180:0:.2);
    \draw (-0.4,-0.8) -- (0.7,-0.8);
    \draw (-0.4,-0.7) -- (0.7,-0.7);
    \filldraw[fill=white, draw=black] (0.1,-0.9) rectangle (0.2, -0.6);
\end{tikzpicture}
}

\newcommand{\threeVAbasisIII}{
\begin{tikzpicture}
    \draw (-0.5,-0.8) arc (150:-30:.2);
    \draw (0,0) -- (0.7,-0.8);
    \draw (0.05,0.05) -- (0.75,-0.75);
    \filldraw[fill=white, draw=black, rotate=-45] (.5,0.15) rectangle (.6, -0.15);
\end{tikzpicture}
}

\newcommand{\threeVAbasisIV}{
\begin{tikzpicture}
    \draw (0.5,-1) arc (210:30:.2);
    \draw (0.3,0) -- (-0.4,-0.8);
    \draw (0.25,0.05) -- (-.45,-0.75);
    \filldraw[fill=white, draw=black, rotate=45] (-.25,-0.05) rectangle (-.35, -0.35);
\end{tikzpicture}
}

\newcommand{\threeVAbasisV}{
\begin{tikzpicture}
    \draw (0,-0.5) arc (-60:0:1);
    \draw (0.6,0.37) arc (180:240:1);
    \draw (0.05,-0.6) arc (120:60:1);
    \filldraw[fill=white, draw=black] (0.7,0.1) rectangle (0.4, 0.2);
    \filldraw[fill=white, draw=black, rotate=30] (-.15,-0.34) rectangle (-.05, -0.64);
    \filldraw[fill=white, draw=black, rotate=-30] (1,0.21) rectangle (1.1, -0.09);
\end{tikzpicture}
}

\newcommand{\threeVAcoeffI}{
\begin{tikzpicture}
    \draw (0,-0.5) arc (-60:0:1);
    \draw (0.6,0.37) arc (180:240:1);
    \draw (0.05,-0.6) arc (120:60:1);
    \draw (0.5,0.37) arc (180:0:.05);
    \draw (0,-0.505) arc (-270:-30:.057);
    \draw (1.05, -0.6) arc (-120:60:.057);
\end{tikzpicture}
}

\newcommand{\LQGarea}{
\begin{tikzpicture}[scale = 0.8]
    \draw[ultra thick] (0,0) -- (0,1.5)
    (0,0) -- (1,-1)
    (0,0) -- (-1,-1)
    (0,0) -- (1.5,0)
    ;
    \draw[Blue, thick] (1.5,1) arc (120:240:1);
    \filldraw[black] (0, 0) circle (0.1cm);
    \node at (1.25, 1) {\color{Blue} $\gamma$};
    \node at (3.5, 0) {$= \sum_j \sqrt{j(j+1)}$};
    \draw[ultra thick] (6,0) -- (6,1.5)
    (6,0) -- (7,-1)
    (6,0) -- (5,-1)
    (6,0) -- (7.5,0)
    ;
    \filldraw[black] (6, 0) circle (0.1cm);
    \node at (6.75, 0.25) {$j$};
\end{tikzpicture}
}

\newcommand{\LQGvolume}{
\begin{tikzpicture}
    \filldraw[fill=Blue!20, draw=Blue, dashed] (0.25,-2.5) -- (1.75,-2.5) -- (1, -3.75) -- (0.25,-2.5);
    \node at (1, -4) {\color{Blue} $S_v$};
    \node at (1.25, -2.25) {$j$};
    \node at (0, -3.5) {$k$};
    \node at (2, -3.5) {$l$};
    \filldraw[black] (1, -3) circle (0.1cm);
    \draw[ultra thick] 
    (1,-3) -- (2,-4)
    (1,-3) -- (0,-4)
    (1,-3) -- (1,-1.75)
    ;
    \node at (3, -3) {$= a_{jkl}$};
    \node at (5.25, -2.25) {$j$};
    \node at (4, -3.5) {$k$};
    \node at (6, -3.5) {$l$};
    \filldraw[black] (5, -3) circle (0.1cm);
    \draw[ultra thick] 
    (5,-3) -- (6,-4)
    (5,-3) -- (4,-4)
    (5,-3) -- (5,-1.75)
    ;
    \draw[-stealth] (0.8,-3) -- (0.8,-2);
    \draw[-stealth] (0.8,-3.05) -- (0.1,-3.75);
    \draw[-stealth] (1.2,-3.05) -- (1.9,-3.75);
    \node at (0.2, -2) {$\dot e_1(0)$};
\end{tikzpicture}
}

\newcommand{\threeVAplain}[1]{
\begin{tikzpicture}[scale=#1]
    \draw (1.5,-0.5) arc (-60:0:2);
    \draw (3,1.25) arc (180:240:2);
    \draw (1.75,-0.75) arc (120:60:2);
\end{tikzpicture}
}

\newcommand{\threeVAplainII}[1]{
\begin{tikzpicture}[scale=#1]
    \draw (3,1.25) arc (180:240:2);
    \draw (2.5,1) arc (180:240:2);
    \draw (1,-0.75) arc (150:-30:0.5);
\end{tikzpicture}
}

\newcommand{\threeVAspin}[1]{
\begin{tikzpicture}[scale=#1]
    \draw (1.5,-0.5) arc (-60:0:2);
    \draw (3,1.25) arc (180:240:2);
    \draw (1.75,-0.75) arc (120:60:2);
    \filldraw[fill=white, draw=black] (3.25, -0.8) -- (3.75, 0) -- (4, -.15) -- (3.5, -0.9) -- (3.25, -0.8);
\end{tikzpicture}
}

\newcommand{\spinone}{
\begin{tikzpicture}
    \draw (0,0) -- (0,-.3)
    (0.2,0) -- (0.2,-0.3)
    ;
    \filldraw[fill=white, draw=black] (-0.05,-0.2) rectangle (0.25, -0.1);
\end{tikzpicture}
}

\newcommand{\spinj}{
\begin{tikzpicture}
    \draw (0,0) -- (0,-.3)
    (0.35,0) -- (0.35,-0.3)
    ;
    \filldraw[fill=white, draw=black] (-0.05,-0.2) rectangle (0.4, -0.1);
    \node at (0.18, -0.35) {\tiny $j$};
    \node at (0.19, -0.025) {\scalebox{.5}{$\dots$}};
\end{tikzpicture}
}

\newcommand{\traceid}[1]{
\begin{tikzpicture}[scale=0.5]
    \draw (0,0) circle (0.5cm);
    \draw (0,0) circle (0.3cm);
    \filldraw[fill=white, draw=black] (-0.6,0.05) rectangle (-0.2, -0.05);
    \node at (0, 0) {\tiny #1};
\end{tikzpicture}
}

\newcommand{\isometryIL}{
\begin{tikzpicture}[scale=0.5]
    \filldraw[fill=black] (-0.5,-0.5) rectangle node {\color{white} $A$} (0.5,0.5);
    \draw[thick] (0,0.5) -- (0,1)
    (-0.2,-0.5) -- (-0.2,-1)
    (-0.4,-0.5) -- (-0.4,-1)
    (0.4,-0.5) -- (0.4,-1)
    (0,-2) -- (0,-2.5)
    ;
    \node at (0.1, -0.75) {\tiny $...$};
    \filldraw[fill=black] (-0.5,-1) rectangle node {\color{white} $\overline A$} (0.5,-2);
\end{tikzpicture}
}

\newcommand{\isometryIR}{
\begin{tikzpicture}[scale=0.5]
    \draw[thick] (0,1) -- (0,-2.5);
\end{tikzpicture}
}

\newcommand{\isometryIIR}{
\begin{tikzpicture}[scale=0.5]
    \draw[thick] (0,1) -- (0,-2.5);
    \draw[thick] (0.2,1) -- (0.2,-2.5);
\end{tikzpicture}
}

\newcommand{\isometryIIIL}{
\begin{tikzpicture}[scale=0.75]
    \filldraw[fill=black] (-0.5,-0.5) rectangle node {\color{white} $A$} (0.5,0.5);
    \draw (0,0.5) -- (0,1)
    (-0.1,0.5) -- (-0.1,1)
    (0.1,0.5) -- (0.1,1)
    (-0.2,-0.5) -- (-0.2,-1)
    (-0.25,-0.5) -- (-0.25,-1)
    (-0.15,-0.5) -- (-0.15,-1)
    (-0.4,-0.5) -- (-0.4,-1)
    (-0.45,-0.5) -- (-0.45,-1)
    (-0.35,-0.5) -- (-0.35,-1)
    (0.4,-0.5) -- (0.4,-1)
    (0.45,-0.5) -- (0.45,-1)
    (0.35,-0.5) -- (0.35,-1)
    (0,-2) -- (0,-2.5)
    (0.1,-2) -- (0.1,-2.5)
    (-0.1,-2) -- (-0.1,-2.5)
    ;
    \node at (0.1, -0.75) {\tiny $...$};
    \filldraw[fill=black] (-0.5,-1) rectangle node {\color{white} $\overline A$} (0.5,-2);
\end{tikzpicture}
}

\newcommand{\isometryIIIR}{
\begin{tikzpicture}[scale=0.75]
    \draw (0,1) -- (0,-2.5)
    (0.1,1) -- (0.1,-2.5)
    (-0.1,1) -- (-0.1,-2.5)
    ;
\end{tikzpicture}
}

\newcommand{\isometryIVL}{
\begin{tikzpicture}[scale=0.75]
    \draw (0.5,0) -- (2,0)
    (0.5,-1.5) -- (2,-1.5)
    (0.5,0.1) -- (2,0.1)
    (0.5,-1.4) -- (2,-1.4)
    (0.5,-0.1) -- (2,-0.1)
    (0.5,-1.6) -- (2,-1.6)
    ;
    \filldraw[fill=black] (-0.5,-0.5) rectangle node {\color{white} $A$} (0.5,0.5);
    \filldraw[fill=black] (2,-0.5) rectangle node {\color{white} $A$} (3,0.5);
    \filldraw[fill=Green!75, fill opacity=1] (1.25,0) circle (0.2cm) node[anchor=north, align=center] {\\[-1ex] \small $B$};
    \filldraw[fill=white] (0.75,0) circle (0.05cm);
    \filldraw[fill=white] (1.75,0) circle (0.05cm);
    \filldraw[fill=white] (0.75,0.1) circle (0.05cm);
    \filldraw[fill=white] (1.75,0.1) circle (0.05cm);
    \filldraw[fill=white] (0.75,-0.1) circle (0.05cm);
    \filldraw[fill=white] (1.75,-0.1) circle (0.05cm);
    \draw (0,0.5) -- (0,1)
    (0.1,0.5) -- (0.1,1)
    (-0.1,0.5) -- (-0.1,1)
    (-0.4,-0.5) -- (-0.4,-1)
    (-0.45,-0.5) -- (-0.45,-1)
    (-0.35,-0.5) -- (-0.35,-1)
    (0.4,-0.5) -- (0.4,-1)
    (0.45,-0.5) -- (0.45,-1)
    (0.35,-0.5) -- (0.35,-1)
    (0,-2) -- (0,-2.5)
    (0.1,-2) -- (0.1,-2.5)
    (-0.1,-2) -- (-0.1,-2.5)
    (2.5,0.5) -- (2.5,1)
    (2.6,0.5) -- (2.6,1)
    (2.4,0.5) -- (2.4,1)
    (2.1,-0.5) -- (2.1,-1)
    (2.15,-0.5) -- (2.15,-1)
    (2.05,-0.5) -- (2.05,-1)
    (2.9,-0.5) -- (2.9,-1)
    (2.95,-0.5) -- (2.95,-1)
    (2.85,-0.5) -- (2.85,-1)
    (2.5,-2) -- (2.5,-2.5)
    (2.6,-2) -- (2.6,-2.5)
    (2.4,-2) -- (2.4,-2.5)
    ;
    \node at (0, -0.75) {\tiny $...$};
    \node at (2.5, -0.75) {\tiny $...$};
    \filldraw[fill=black] (-0.5,-1) rectangle node {\color{white} $\overline A$} (0.5,-2);
    \filldraw[fill=black] (3,-1) rectangle node {\color{white} $\overline A$} (2,-2);
    \filldraw[fill=Green!75, fill opacity=1] (1.25,-1.5) circle (0.2cm) node[anchor=north, align=center] {\\[-1ex] \small $\overline B$};
    \filldraw[fill=white] (0.75,-1.5) circle (0.05cm);
    \filldraw[fill=white] (0.75,-1.4) circle (0.05cm);
    \filldraw[fill=white] (0.75,-1.6) circle (0.05cm);
    \filldraw[fill=white] (1.75,-1.5) circle (0.05cm);
    \filldraw[fill=white] (1.75,-1.6) circle (0.05cm);
    \filldraw[fill=white] (1.75,-1.4) circle (0.05cm);
\end{tikzpicture}
}

\newcommand{\isometryIVR}{
\begin{tikzpicture}[scale=0.75]
    \draw (0,1) -- (0,-2.5)
    (0.1,1) -- (0.1,-2.5)
    (0.2,1) -- (0.2,-2.5)
    (0.6,1) -- (0.6,-2.5)
    (0.7,1) -- (0.7,-2.5)
    (0.5,1) -- (0.5,-2.5)
    ;
\end{tikzpicture}
}

\newcommand{\starA}{
\begin{tikzpicture}[scale=0.7]
    \draw (-1,-1) -- (1,1)
    (1,-1) -- (-1,1)
    (0,-1) -- (0,1)
    (-1,0) -- (1,0)
    ;
\end{tikzpicture}
}

\newcommand{\LRA}{
\begin{tikzpicture}[scale=0.7]
    \draw (1.5,-0.5) arc (-60:0:2);
    \draw (3,1.25) arc (180:240:2);
    \draw (1.75,-0.75) arc (120:60:2);
\end{tikzpicture}
}

\newcommand{\LRAII}{
\begin{tikzpicture}[scale=0.7]
    \draw (0,0) arc (45:-45:1);
    \draw (0.2,0.1) arc (-135:-45:1);
    \draw (0.2,-1.6) arc (135:45:1);
    \draw (1.8,0) arc (135:225:1);
\end{tikzpicture}
}

\newcommand{\LRAIII}{
\begin{tikzpicture}[scale=0.7]
    \draw (-1,-1) -- (1,1)
    (1,-1) -- (-1,1)
    ;
    \draw (-0.75,1) arc (-135:-45:1);
    \draw (1.2,-1) arc (-135:-180:3.9);
    \draw (-1.2,-1) arc (-45:0:3.9);
\end{tikzpicture}
}

\newcommand{\LRAIV}{
\begin{tikzpicture}[scale=0.7]
    \draw (-1,-1) -- (1,1)
    (1,-1) -- (-1,1)
    ;
    \draw (-0.7,1) arc (-135:-45:1);
    \draw (-0.7,-1) arc (135:45:1);
    \draw (1.2,-1) arc (-135:-173:3.9);
    \draw (-1.2,-1) arc (-45:-7:3.9);
    \draw (-1.1,-0.6) arc (-45:45:1);
    \draw (1.1,-0.6) arc (-135:-225:1);
    \draw (-0.6,1) arc (-135:0:.25);
    \draw (0.6,1) arc (-45:-180:.25);
\end{tikzpicture}
}

\newcommand{\LRAV}{
\begin{tikzpicture}[scale=0.7]
    \draw (1.5,1) -- (3.5,-1)
    (3.5,1) -- (1.5,-1)
    ;
    \draw (1.6,1) arc (-120:-60:1.8);
    \draw (1.6,-1) arc (120:60:1.8);
    \draw (1.5,0.9) arc (30:-30:1.8);
    \draw (3.5,0.9) arc (150:210:1.8);
\end{tikzpicture}
}

\newcommand{\LRBV}{
\begin{tikzpicture}[scale=0.7]
    \draw (1.5,0.5) -- (2,0.5) -- (3,-0.5) -- (3.5,-0.5)
    (1.5,-0.5) -- (2,-0.5) -- (3,0.5) -- (3.5,0.5)
    (1.5,0) -- (3.5,0)
    ;
\end{tikzpicture}
}

\newcommand{\LRBIV}{
\begin{tikzpicture}[scale=0.7]
    \draw (1.5,0.5) -- (2,0.5) -- (3,-0.5) -- (3.5, -0.5)
    (1.5,-0.5) -- (2,-0.5) -- (3,0.5) -- (3.5,0.5)
    (1.5,0.1) -- (3.5,0.1)
    (1.5,-0.1) -- (3.5,-0.1)
    ;
\end{tikzpicture}
}

\newcommand{\LRBIII}{
\begin{tikzpicture}[scale=0.7]
    \draw (1.5,0.25) -- (3,0.25)
    (1.5,-0.25) -- (3,-0.25)
    ;
\end{tikzpicture}
}

\newcommand{\LRBI}{
\begin{tikzpicture}[scale=0.7]
    \draw (1,0.25) -- (1.5,0.25) -- (3,-0.25) -- (3.5,-0.25)
    (1,-0.25) -- (1.5,-0.25) -- (3,0.25) -- (3.5,0.25)
    ;
\end{tikzpicture}
}

\newcommand{\binorI}{
\begin{tikzpicture}[scale=0.2]
    \draw (-1,1) -- (1,-1)
    (-1,-1) -- (1,1);
\end{tikzpicture}
}

\newcommand{\binorII}{
\begin{tikzpicture}[scale=0.2]
    \draw (0,1) arc (-180:0:0.8);
    \draw (0,-1) arc (180:0:0.8);
\end{tikzpicture}
}

\newcommand{\SNSdisc}{
    \begin{tikzpicture}
    \node[anchor=south west,inner sep=0] at (0,0) {\includegraphics[width=\linewidth]{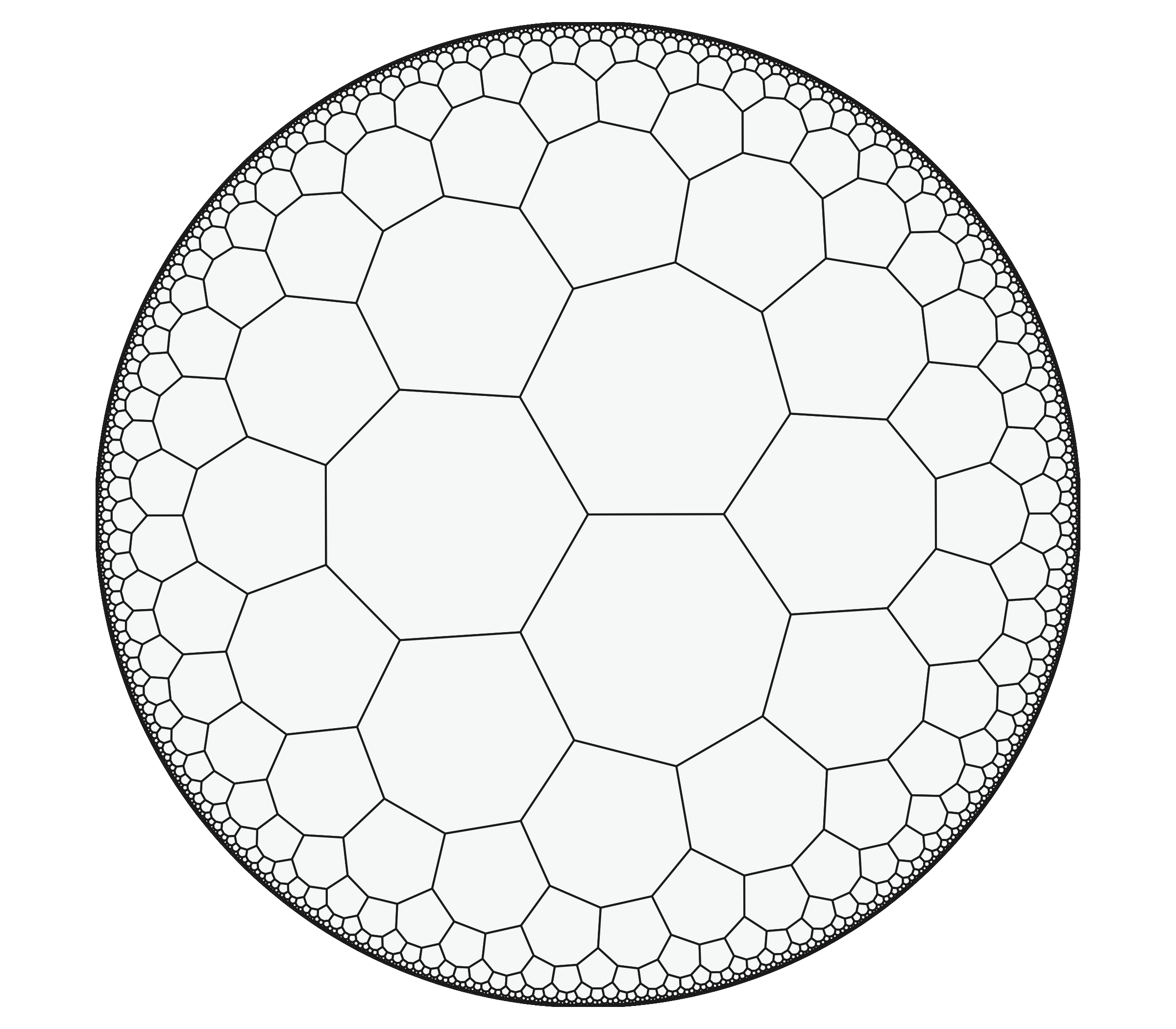}};
    \fill (4.33,3.79) circle (0.15cm);
    \fill[Green] (4.85,3.79) circle (0.1cm);
    \fill[Green] (4.1,4.2) circle (0.1cm);
    \fill[Green] (4.1,3.35) circle (0.1cm);
    \fill (5.35,3.78) circle (0.15cm);
    \fill (3.85,2.93) circle (0.15cm);
    \fill (3.85,4.65) circle (0.15cm);
    \node at (4.45,3.5) {\small $A$};
    \node at (4.85,3.5) {\color{Green} \small $B$};
\end{tikzpicture}
}

\newcommand{\isometrydouble}{
\begin{tikzpicture}[scale=0.5]
    \draw[thick] (0.5,0) -- (2,0)
    (0.5,-1.5) -- (2,-1.5)
    ;
    \filldraw[fill=black] (-0.5,-0.5) rectangle node {\color{white} $A$} (0.5,0.5);
    \filldraw[fill=black] (2,-0.5) rectangle node {\color{white} $A$} (3,0.5);
    \filldraw[fill=Green!75, fill opacity=1] (1.25,0) circle (0.5cm) node {$B$};
    \draw[thick] (0,0.5) -- (0,1)
    (-0.4,-0.5) -- (-0.4,-1)
    (-0.2,-0.5) -- (-0.2,-1)
    (0.4,-0.5) -- (0.4,-1)
    (0,-2) -- (0,-2.5)
    (2.5,0.5) -- (2.5,1)
    (2.3,-0.5) -- (2.3,-1)
    (2.1,-0.5) -- (2.1,-1)
    (2.9,-0.5) -- (2.9,-1)
    (2.5,-2) -- (2.5,-2.5)
    ;
    \node at (0.1, -0.75) {\tiny $...$};
    \node at (2.6, -0.75) {\tiny $...$};
    \filldraw[fill=black] (-0.5,-1) rectangle node {\color{white} $\overline A$} (0.5,-2);
    \filldraw[fill=black] (3,-1) rectangle node {\color{white} $\overline A$} (2,-2);
    \filldraw[fill=Green!75, fill opacity=1] (1.25,-1.5) circle (0.5cm) node {$\overline B$};
\end{tikzpicture}
}

\newcommand{\isometrysingleB}{
\begin{tikzpicture}[scale=0.5]
    \draw[thick] (0,1) -- (0,-2.5);
    \filldraw[fill=Green!75, fill opacity=1] (0,0) circle (0.5cm) node {$B$};
    \filldraw[fill=Green!75, fill opacity=1] (0,-1.5) circle (0.5cm) node {$\overline B$};
\end{tikzpicture}
}

\newcommand{\holonomy}{
\begin{tikzpicture}
    \draw[-stealth, blue] (0,0) -- (-0.2,1);
    \draw[-stealth, blue] (4,0) -- (4.7,0.7);
    \fill (0,0) circle (0.1cm) node[anchor=north] {$e(0)$};
    \fill (4,0) circle (0.1cm) node[anchor=north] {$e(1)$};
    \draw[decorate, decoration={snake, segment length=100, amplitude=10}] (0,0) -- (4,0);
    \node at (2.5,-0.65) {$\pi_j(h_e[A])$};
    \node at (1,0.6) {$j$};
    \node at (0.3,0.8) {\color{blue} $v(0)$};
    \node at (4,0.6) {\color{blue} $v(1)$};
\end{tikzpicture}
}

\newcommand{\intertwiner}{
\begin{tikzpicture}
    \draw[-stealth, blue] (0,0) -- (-0.2,1);
    \draw[-stealth, Green] (5.4,0.6) -- (5.1,0.9);
    \draw[-stealth, Green] (5.4,-0.6) -- (5.7,-.3);
    \fill (0,0) circle (0.1cm);
    \fill (4,0) circle (0.1cm);
    \draw[decorate, decoration={snake, segment length=100, amplitude=10}] (0,0) -- (4,0);
    \draw (4,0) arc (-90:-45:2);
    \draw (4,0) arc (90:45:2);
    \node at (0.1,-0.25) {\tiny $\psi(p_1)$};
    \node at (3.6,0.3) {\tiny $I^k_{\mu_2 \mu_3}$};
    \node at (1,0.6) {\tiny $j=1, h_{e_1}$};
    \node at (4.5,0.5) {\tiny $j=\frac{1}{2}, h_{e_2}$};
    \node at (4.5,-0.5) {\tiny $j=\frac{1}{2}, h_{e_3}$};
\end{tikzpicture}
}

\newcommand{\threeVal}{
\begin{tikzpicture}[scale=0.7]
    \node at (1.25, -2.25) {$j$};
    \node at (0, -3.5) {$k$};
    \node at (2, -3.5) {$l$};
    \filldraw[black] (1, -3) circle (0.1cm);
    \draw[ultra thick] 
    (1,-3) -- (2,-4)
    (1,-3) -- (0,-4)
    (1,-3) -- (1,-1.75)
    ;
\end{tikzpicture}
}

\newcommand{\grasping}{
\begin{tikzpicture}
    \node at (1.75, -2.25) {$1$};
    \node at (2.25, -2.5) {\small $i$};
    \node at (1.25, -2) {$j$};
    \node at (1.25, -2.75) {\small $j$};
    \node at (0, -3.5) {$k$};
    \node at (2, -3.5) {$l$};
    \filldraw[black] (1, -3) circle (0.1cm);
    \filldraw[black] (1, -2.5) circle (0.05cm);
    \draw[ultra thick] 
    (1,-3) -- (2,-4)
    (1,-3) -- (0,-4)
    (1,-3) -- (1,-1.75)
    ;
    \draw (1,-2.5) -- (2,-2.5);
\end{tikzpicture}
}

\newcommand{\graspingEps}{
\begin{tikzpicture}[scale=0.7]
    \node at (1.5, -2.5) {$1$};
    \node at (2, -2.7) {\small $i_3$};
    \node at (1.25, -2) {$j$};
    \node at (0.8, -2.75) {\tiny $j$};
    \node at (0, -3.5) {$k$};
    \node at (2, -3.5) {$l$};
    \node at (1.2, -3.4) {\tiny $l$};
    \filldraw[black] (1, -3) circle (0.1cm);
    \filldraw[black] (1, -2.5) circle (0.05cm);
    \filldraw[black] (1.5, -3.5) circle (0.05cm);
    \filldraw[black] (1.57, -2.85) circle (0.05cm);
    \draw[ultra thick] 
    (1,-3) -- (2,-4)
    (1,-3) -- (0,-4)
    (1,-3) -- (1,-1.75)
    ;
    \draw (1,-2.5) arc (90:-33:0.63);
    \draw (1.57, -2.85) -- (1.8, -2.7);
\end{tikzpicture}
}

\newcommand{\graspingEpsII}{
\begin{tikzpicture}
    \node at (1.5, -2.5) {\small $1$};
    \node at (.5, -2.5) {\small $1$};
    \node at (1.25, -2.2) {$j$};
    \node at (0.8, -2.75) {\tiny $j$};
    \node at (0, -3.5) {$k$};
    \node at (2, -3.5) {$l$};
    \node at (1.2, -3.4) {\tiny $l$};
    \node at (0.8, -3.4) {\tiny $k$};
    \filldraw[black] (1, -3) circle (0.1cm);
    \filldraw[black] (1, -2.5) circle (0.05cm);
    \filldraw[black] (1.5, -3.5) circle (0.05cm);
    \filldraw[black] (.5, -3.5) circle (0.05cm);
    \filldraw[black] (1.57, -2.85) circle (0.05cm);
    \filldraw[black] (0.43, -2.85) circle (0.05cm);
    \draw[ultra thick] 
    (1,-3) -- (2,-4)
    (1,-3) -- (0,-4)
    (1,-3) -- (1,-2)
    ;
    \draw (1,-2.5) arc (90:-33:0.63);
    \draw (1,-2.5) arc (90:213:0.63);
    \draw (0.43,-2.85) arc (213:-33:0.69);
\end{tikzpicture}
}

\newcommand{\SNS}{
\begin{tikzpicture}
    \node[anchor=south west,inner sep=0] at (0,0) {\includegraphics[scale=0.5]{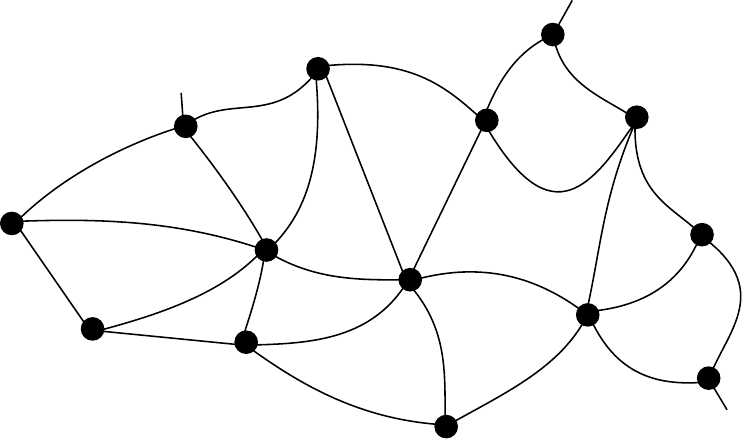}};
    \node at (0.4,1.2) {$j_3$};
    \node at (1,2) {$j_4$};
    \node at (1,2.7) {$j_5$};
    \node at (1.6,1.4) {$j_2$};
    \node at (1.5,0.6) {$j_1$};
    \node at (2.2,3) {$...$};
    \node at (1.7,2.9) {\tiny $j$};
    \node at (1.6,3.2) {$\theta^{\mu_j}$};
    \node at (4.9,3.6) {\tiny $k$};
    \node at (5,3.9) {$\theta^{\mu_k}$};
    \node at (6,0.3) {\tiny $l$};
    \node at (6.2,0) {$\theta^{\mu_l}$};
\end{tikzpicture}
}

\newcommand{\isometryExample}{
\begin{tikzpicture}
    \draw (-0.5,1) -- (-0.5,-1);
    \draw (0.5,1) arc (180:270:0.5);
    \draw (0.5,-1) arc (180:90:0.5);
    \draw (1,0.4) arc (90:270:0.4);
    \draw (-0.4,0.5) arc (180:0:0.1);
    \node at (0, 0.5) {$...$};
    \draw (0.2,0.5) arc (180:0:0.1);
    \draw (-0.4,-0.5) arc (-180:0:0.1);
    \node at (0, -0.5) {$...$};
    \draw (0.2,-0.5) arc (-180:0:0.1);
    \draw (-0.4,0.5) -- (-0.4,-0.5)
    (-0.2,0.5) -- (-0.2,-0.5)
    (0.2,0.5) -- (0.2,-0.5)
    (0.4,0.5) -- (0.4,-0.5)
    ;
    \draw (1,0.5) -- (1.1,0.4);
    \draw (1,0.4) -- (1.1,0.5);
    \draw (1,-0.5) -- (1.1,-0.4);
    \draw (1,-0.4) -- (1.1,-0.5);
    \draw (2.6,1) -- (2.6,-1);
    \draw (1.6,1) arc (0:-90:0.5);
    \draw (1.6,-1) arc (0:90:0.5);
    \draw (1.1,0.4) arc (90:-90:0.4);
    \draw (1.7,0.5) arc (180:0:0.1);
    \node at (2.1, 0.5) {$...$};
    \draw (2.3,0.5) arc (180:0:0.1);
    \draw (1.7,-0.5) arc (-180:0:0.1);
    \node at (2.1, -0.5) {$...$};
    \draw (2.3,-0.5) arc (-180:0:0.1);
    \draw (1.7,0.5) -- (1.7,-0.5)
    (1.9,0.5) -- (1.9,-0.5)
    (2.5,0.5) -- (2.5,-0.5)
    (2.3,0.5) -- (2.3,-0.5)
    ;
\end{tikzpicture}
}

\newcommand{\lengthvariance}{
\begin{tikzpicture}[scale=0.7]
    \draw (-1,0.9) -- (3,0.9)
    (-1,0.7) -- (3,0.7)
    (-1,0.5) -- (3,0.5)
    (-1,0.3) -- (3,0.3)
    (-1,0.1) -- (3,0.1)
    (-1,-0.1) -- (3,-0.1)
    (-1,-0.3) -- (3,-0.3)
    (-1,-0.5) -- (3,-0.5)
    (-1,-0.7) -- (3,-0.7)
    (-1,-0.9) -- (3,-0.9)
    (-1,-1.1) -- (3,-1.1)
    (-1,-1.3) -- (3,-1.3)
    (-1,-1.5) -- (3,-1.5)
    (-1,-1.7) -- (3,-1.7)
    ;
    \filldraw[fill = blue!20] (0,1) arc (45:-45:2)
    (0,1) arc (135:225:2);
    \filldraw[fill = blue!20] (2,1) arc (45:-45:2)
    (2,1) arc (135:225:2);
    \filldraw[fill=white] (0.8,0.4) rectangle (1.2,0.8);
    \filldraw[fill=white] (0.8,-0.4) rectangle (1.2,-0.8);
    \node at (0,-0.3) {$\mathcal{E}_\mathcal{A}$};
    \node at (2,-0.3) {$\overline{\mathcal{E}_\mathcal{A}}$};
    \node at (1,0.6) {\scalebox{.5}{$P_x^j$}};
    \node at (1,-0.6) { \scalebox{.5}{$P_y^k$}};
\end{tikzpicture}
}

\newcommand{\lengthvariancestrip}{
\begin{tikzpicture}[scale=0.7]
    \draw (-1,0.9) -- (3,0.9)
    (-1,0.7) -- (3,0.7)
    (-1,0.5) -- (3,0.5)
    (-1,0.3) -- (3,0.3)
    (-1,0.1) -- (3,0.1)
    (-1,-0.1) -- (3,-0.1)
    (-1,-0.3) -- (3,-0.3)
    (-1,-0.5) -- (3,-0.5)
    (-1,-0.7) -- (3,-0.7)
    (-1,-0.9) -- (3,-0.9)
    (-1,-1.1) -- (3,-1.1)
    (-1,-1.3) -- (3,-1.3)
    (-1,-1.5) -- (3,-1.5)
    (-1,-1.7) -- (3,-1.7)
    ;
    \filldraw[fill=white] (0.8,0.4) rectangle (1.2,0.8);
    \filldraw[fill=white] (0.8,-0.4) rectangle (1.2,-0.8);
    \filldraw[fill = blue!20] (0,1) arc (45:-45:2) (0,1) -- (0,-1.8);
    \filldraw[fill = blue!20] (2,1) -- (2,-1.8) (2,1) arc (135:225:2);
    \node at (0.3,-0.3) {\scalebox{.7}{$\sigma_\mathcal{A}$}};
    \node at (1.7,-0.3) {\scalebox{.7}{$\overline{\sigma_\mathcal{A}}$}};
    \node at (1,0.6) {\scalebox{.5}{$P_x^j$}};
    \node at (1,-0.6) {\scalebox{.5}{$P_y^k$}};
\end{tikzpicture}
}

\newcommand{\correlations}{
\begin{tikzpicture}
    \filldraw[fill=white] (0,0) rectangle (1,1);
    \node at (0.5,0.5) {$O_1$};
    \filldraw[fill=white] (2,0) rectangle (3,1);
    \node at (2.5,0.5) {$O_2$};
    \draw (0.5,1) arc (110:70:3)
    (0.5,0) arc (-110:-70:3)
    ;
\end{tikzpicture}
}

\newcommand{\correlationsII}{
\begin{tikzpicture}
    \draw (0.75,1) arc (110:70:2.25)
    (0.75,0) arc (-110:-70:2.25)
    (0.25,1) arc (45:315:0.7)
    (2.75,1) arc (135:-135:0.7)
    ;
    \filldraw[fill=white] (0,0) rectangle (1,1);
    \node at (0.5,0.5) {$O_1$};
    \filldraw[fill=white] (2,0) rectangle (3,1);
    \node at (2.5,0.5) {$O_2$};
\end{tikzpicture}
}

\newcommand{\SNSdiscStrip}{
    \begin{tikzpicture}[scale=1.4]
    \node[anchor=south west,inner sep=0] at (0,0) {\includegraphics[scale=0.42]{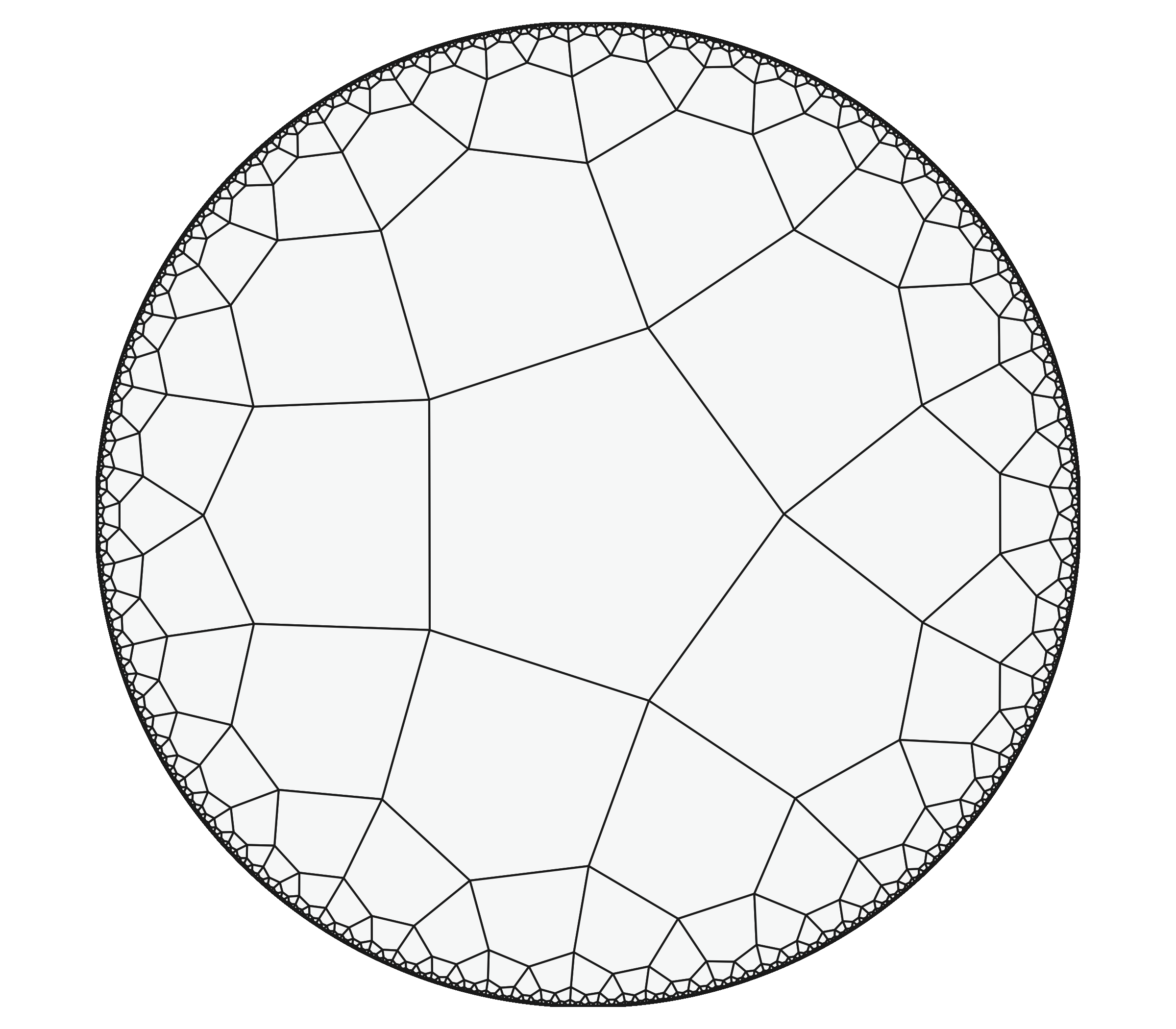}};
    \draw[blue] (1.4,4.75) -- (1.75,3.5) -- (1,1.5) -- (1.35, 1.1) -- (1,0.7);
    \draw[Green] (1.4,4.75) -- (1.55, 4.3) -- (2, 4.3) -- (2.5,3) -- (2.5, 1) -- (1.4,1.1) -- (1,0.7);
    \node at (1.9,2.75) {$\sigma_\mathcal{A}$};
    \node at (1.5,5) {\color{Green}$\gamma_{\mathcal{A}^c}$};
    \node at (1,4.75) {\color{blue} $\gamma_\mathcal{A}$};
    \fill (1.97,4.15) circle (0.05cm);
    \fill (2.22,3.27) circle (0.05cm);
    \fill (2.22,2.07) circle (0.05cm);
    \fill (1.98,1.21) circle (0.05cm);
    \fill (1.44,1.24) circle (0.05cm);
    \fill (1.22,1.58) circle (0.05cm);
    \fill (1.32,2.1) circle (0.05cm);
    \fill[Green] (1.75,2.09) circle (0.05cm);
    \fill[Green] (2.1, 3.72) circle (0.05cm);
    \fill[Green] (2.23, 2.65) circle (0.05cm);
    \fill[Green] (1.26, 1.85) circle (0.05cm);
    \fill[Green] (1.33, 1.4) circle (0.05cm);
    \fill[Green] (1.7, 1.22) circle (0.05cm);
    \fill[Green] (2.1, 1.62) circle (0.05cm);
\end{tikzpicture}
}

\newcommand{\anglesum}{
\begin{tikzpicture}
     \node[anchor=south west,inner sep=0] at (0,0) {\includegraphics[scale=0.25]{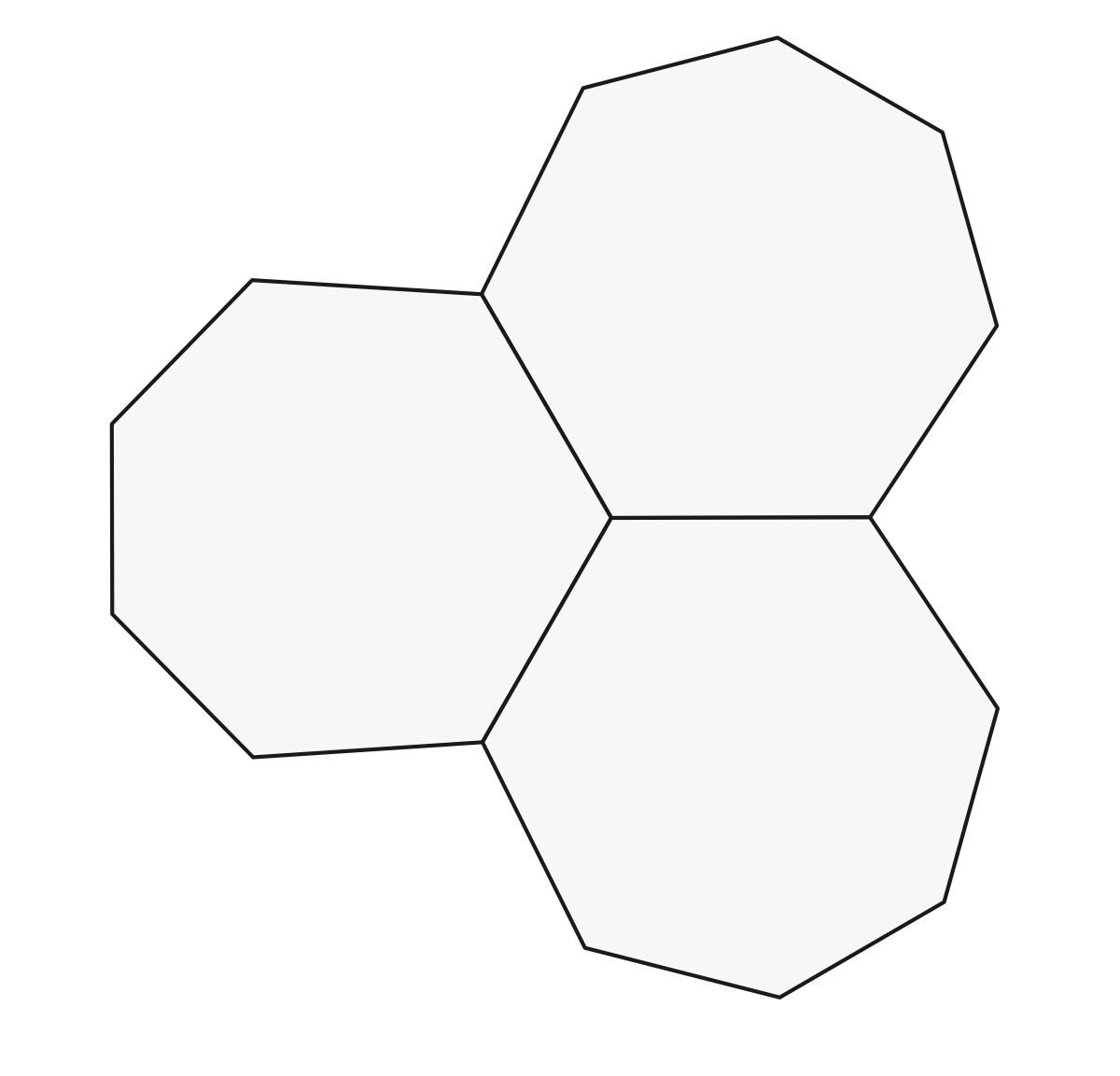}};
     \draw[Red] (0.5,2) -- (0,3) -- (0.75,4.5) -- (2.3,5.4) -- (3.75,6.2) -- (5.1,6.1) -- (6,5) -- (6,3.25) -- (6,1.5) -- (5.1, 0.3) -- (3.75,0.2) -- (2.3,1) -- (0.5,2);
     \draw[Blue] (2,3) -- (4.2,4.5) -- (4.2,1.8) -- (2,3)
     ;
     \draw[Blue] (4.2,4.5) -- (2.3,5.4)
     (4.2,4.5) -- (3.75,6.2)
     (4.2,4.5) -- (5.1,6.1)
     (4.2,4.5) -- (6,5)
     (4.2,4.5) -- (6,3.25)
     (2,3) -- (2.3,5.4)
     (4.2,1.8) -- (6,3.25)
     ;
     \draw[Orange] (4.01,4.35) arc (210:270:0.2);
     \draw[Orange] (3.6,6.1) arc (210:350:0.2);
     \draw[Orange] (2.12,5.29) arc (-150:30:0.2);
     \node at (4, 4.1) {\color{Orange}$\alpha$};
     \node at (3.8, 6.5) {\color{Orange}$2\alpha$};
     \node at (2, 5.5) {\color{Orange}$3\alpha$};
     \fill[Red] (2.3,1) circle (0.05cm)
     (0.5,2) circle (0.05cm)
     (0,3) circle (0.05cm)
     (0.75,4.5) circle (0.05cm)
     (2.3,5.4) circle (0.05cm)
     (3.75,6.2) circle (0.05cm)
     (5.1,6.1) circle (0.05cm)
     (6,5) circle (0.05cm)
     (6,3.25) circle (0.05cm)
     (6,1.5) circle (0.05cm)
     (5.1,0.3) circle (0.05cm)
     (3.75,0.2) circle (0.05cm)
     ;
\end{tikzpicture}
}

\newcommand{\graspedA}[1]{
\begin{tikzpicture}[rotate=#1]
    \draw (0,-0.5) arc (-60:-10:1);
    \draw (0.49,0.19) -- (0.25,0.19)
    (0.49, 0.37) -- (0.49,0.25) -- (0.25,0.25);
    \filldraw[fill=white] (0.325,0.15) rectangle (0.4,0.3);
    \draw (0.6,0.37) arc (180:240:1);
    \draw (0.05,-0.6) arc (120:60:1);
\end{tikzpicture}
}

\newcommand{\graspedAII}[1]{
\begin{tikzpicture}[rotate=#1]
    \draw (0,-0.5) arc (-60:0:1);
    \draw (0.6,0.2) arc (190:240:1);
    \draw (0.6,0.2) -- (0.85,0.2)
    (0.6, 0.37) -- (0.6,0.25) -- (0.85,0.25);
    \filldraw[fill=white] (0.675,0.15) rectangle (0.75,0.3);
    \draw (0.05,-0.6) arc (120:60:1);
\end{tikzpicture}
}

\newcommand{\dualtriangle}{
\begin{tikzpicture}
    \filldraw[fill=Blue!20, draw=Blue, dashed] (0.25,-2.5) -- (1.75,-2.5) -- (1, -3.8) -- (0.25,-2.5);
    \draw[ultra thick] 
    (1,-3) -- (2,-4)
    (1,-3) -- (0,-4)
    (1,-3) -- (1,-1.75)
    ;
    \draw[-stealth] (1.75,-2.5) -- (1.75,-1.75);
    \draw[-stealth] (1.75,-2.5) -- (2.25,-2.9);
    \draw[Blue, thick] (1.5, -2.5) arc (180:240:0.2);
    \draw[Blue, thick] (1.75, -2.25) arc (90:-35:0.25);
    \draw[Blue, thick] (1, -2.75) arc (90:-35:0.25);
    \node at (1.75,-2.75) {$\alpha$};
    \node at (2.1,-2.25) {$\theta$};
\end{tikzpicture}
}

\newcommand{\Pzero}{
\begin{tikzpicture}
    \draw (0,0) arc (0:180:0.1);
    \draw (0,0.3) arc (0:-180:0.1);
\end{tikzpicture}
}

\newcommand{\Pone}{
\begin{tikzpicture}
    \draw (0,0) -- (0,0.4)
    (0.2,0) -- (0.2,0.4)
    ;
    \filldraw[fill=white] (-0.1, 0.125) rectangle (0.3, 0.275);
\end{tikzpicture}
}

\newcommand{\PonehalfI}{
\begin{tikzpicture}
    \draw (0,0) arc (0:180:0.1);
    \draw (0,0.3) arc (0:-180:0.1);
    \draw (0.3, 0) -- (0.3,0.3);
\end{tikzpicture}
}

\newcommand{\PonehalfII}{
\begin{tikzpicture}
    \draw (0,0) -- (0,0.7)
    (0.2,0) -- (0.2,0.4)
    (0.4,0) -- (0.4,0.4)
    ;
    \filldraw[fill=white] (-0.1, 0.125) rectangle (0.3, 0.275);
    \draw (0.2, 0.4) arc (180:0:0.1);
    \draw (0.2, 0.7) arc (-180:0:0.1);
\end{tikzpicture}
}

\newcommand{\Pthreehalf}{
\begin{tikzpicture}
    \draw (0,0) -- (0,0.7)
    (0.2,0) -- (0.2,0.7)
    (0.4,0) -- (0.4,0.7)
    ;
    \filldraw[fill=white] (-0.1, 0.125) rectangle (0.3, 0.275);
    \filldraw[fill=white] (-0.1, 0.375) rectangle (0.5, 0.525);
\end{tikzpicture}
}

\newcommand{\PthreehalfII}{
\begin{tikzpicture}
    \draw (0,0.2) -- (0,0.7)
    (0.2,0.2) -- (0.2,0.7)
    (0.4,0.2) -- (0.4,0.7)
    ;
    \filldraw[fill=white] (-0.1, 0.375) rectangle (0.5, 0.525);
\end{tikzpicture}
}

\newcommand{\pushing}{
\begin{tikzpicture}
    \filldraw[fill=black] (-0.5,-0.5) rectangle node {\color{white} $A$} (0.5,0.5);
    \draw (0,0.5) -- (0,1)
    (-0.1,0.5) -- (-0.1,1)
    (0.1,0.5) -- (0.1,1)
    (-0.2,-0.5) -- (-0.2,-1)
    (-0.25,-0.5) -- (-0.25,-1)
    (-0.15,-0.5) -- (-0.15,-1)
    (-0.4,-0.5) -- (-0.4,-1)
    (-0.45,-0.5) -- (-0.45,-1)
    (-0.35,-0.5) -- (-0.35,-1)
    (0.4,-0.5) -- (0.4,-1)
    (0.45,-0.5) -- (0.45,-1)
    (0.35,-0.5) -- (0.35,-1)
    ;
    \node at (0.1, -0.75) {\tiny $...$};
    \filldraw[fill=White] (0, 0.75) circle (0.035cm) node[anchor=east] {$U^{\otimes k}$ \,};
    \filldraw[fill=White] (0.1, 0.75) circle (0.035cm);
    \filldraw[fill=White] (-0.1, 0.75) circle (0.035cm);
\end{tikzpicture}
}

\newcommand{\pushingII}{
\begin{tikzpicture}
    \filldraw[fill=black] (-0.5,-0.5) rectangle node {\color{white} $A$} (0.5,0.5);
    \draw (0,0.5) -- (0,1)
    (-0.1,0.5) -- (-0.1,1)
    (0.1,0.5) -- (0.1,1)
    (-0.2,-0.5) -- (-0.2,-1)
    (-0.25,-0.5) -- (-0.25,-1)
    (-0.15,-0.5) -- (-0.15,-1)
    (-0.4,-0.5) -- (-0.4,-1)
    (-0.45,-0.5) -- (-0.45,-1)
    (-0.35,-0.5) -- (-0.35,-1)
    (0.4,-0.5) -- (0.4,-1)
    (0.45,-0.5) -- (0.45,-1)
    (0.35,-0.5) -- (0.35,-1)
    ;
    \node at (0.1, -0.75) {\tiny $...$};
    \filldraw[fill=White] (0.4, -0.75) circle (0.025cm) node[anchor=west] {$U^{\dagger \otimes k}$};
    \filldraw[fill=White] (0.45, -0.75) circle (0.025cm);
    \filldraw[fill=White] (0.35, -0.75) circle (0.025cm);
    \filldraw[fill=White] (-0.2, -0.75) circle (0.025cm);
    \filldraw[fill=White] (-0.25, -0.75) circle (0.025cm);
    \filldraw[fill=White] (-0.15, -0.75) circle (0.025cm);
    \filldraw[fill=White] (-0.4, -0.75) circle (0.025cm) node[anchor=east] {$U^{\dagger \otimes k}$};
    \filldraw[fill=White] (-0.45, -0.75) circle (0.025cm);
    \filldraw[fill=White] (-0.35, -0.75) circle (0.025cm);
\end{tikzpicture}
}

\newcommand{\pushingB}{
\begin{tikzpicture}
    \draw[thick] (0,0.75) -- (0,-0.75);
    \filldraw[fill=Green!75, fill opacity=1] (0,0) circle (0.3cm) node {$B$};
    \filldraw[fill=White] (0, 0.5) circle (0.05cm) node[anchor=east] {$U$};
\end{tikzpicture}
}

\newcommand{\pushingBIII}{
\begin{tikzpicture}
    \draw[thick] (0,0.75) -- (0,-0.75);
    \filldraw[fill=Green!75, fill opacity=1] (0,0) circle (0.3cm) node {$\tilde B$};
\end{tikzpicture}
}

\begin{document}

\title{Hyperinvariant Spin Network States -- An AdS/CFT Model from First Principles}

\author{Fynn Otto \orcidlink{0009-0009-0382-042X}}
\email[]{fynn.otto@uni-siegen.de}
\affiliation{Naturwissenschaftlich-Technische Fakultät, Universität Siegen,
Walter-Flex-Straße 3, 57068 Siegen, Germany}

\author{Refik Mansuroglu \orcidlink{0000-0001-7352-513X}}
\email[]{refik.mansuroglu@univie.ac.at}
\affiliation{University of Vienna, Faculty of Physics, Boltzmanngasse 5, 1090 Vienna, Austria}

\author{Norbert Schuch \orcidlink{0000-0001-6494-8616}}
\affiliation{University of Vienna, Faculty of Physics, Boltzmanngasse 5, 1090 Vienna, Austria}
\affiliation{University of Vienna, Faculty of Mathematics, Oskar-Morgenstern-Platz 1, 1090 Vienna, Austria}

\author{Otfried Gühne \orcidlink{0000-0002-6033-0867}}
\affiliation{Naturwissenschaftlich-Technische Fakultät, Universität Siegen,
Walter-Flex-Straße 3, 57068 Siegen, Germany}

\author{Hanno Sahlmann \orcidlink{0000-0002-8083-7139}}
\affiliation{Department of Physics, Friedrich-Alexander-Universität Erlangen-Nürnberg (FAU), Staudtstraße 7, 91058 Erlangen, Germany}

\date{\today}

\begin{abstract}
    We study the existence and limitations of hyperinvariant tensor networks incorporating a local $\SUT$ symmetry. As discrete implementations of the anti de-Sitter/conformal field theory (AdS/CFT) correspondence, such networks have created bridges between the fields of quantum information theory and quantum gravity. Adding $\SUT$ symmetry to the tensor network allows a direct connection to spin network states, a basis of the kinematic Hilbert space of loop quantum gravity (LQG). We consider a particular situation where the states can be interpreted as kinematic quantum states for three-dimensional quantum gravity. \change{We demonstrate an entanglement-geometry correspondence realized in certain quantum states of the gravitational field in LQG, showing that models similar to the ones introduced by [F. Pastawski et al., JHEP 06, 149 (2015)] can be obtained as quantum gravitational models from first principles.} We provide examples of hyperinvariant tensor networks, but also prove constraints on their existence in the form of no-go theorems that exclude absolutely maximally entangled states as well as general holographic codes from local $\SUT$-invariance, \change{as long as the logical index also transforms covariantly under the gauge group}. We calculate surface areas as expectation values of the LQG area operator and discuss further possible constraints as a consequence of a decay of correlations on the boundary.
\end{abstract}

\maketitle

\section{Introduction}
The anti-de Sitter/conformal field theory (AdS/CFT) correspondence has emerged as a profound insight into the nature of quantum gravity, offering a concrete realization of the holographic principle \cite{Maldacena_1999, Gubser_1998}. At its core, the AdS/CFT duality posits an equivalence between a gravitational theory in a bulk AdS spacetime and a conformal field theory (CFT) defined on its boundary. This relationship provides a compelling framework for the idea that spacetime geometry, and perhaps even gravity itself, can emerge from patterns of quantum entanglement \cite{swingle2012constructingholographicspacetimesusing, VanRaamsdonk:2010pw, Jacobson_2016}. The correspondence has become a cornerstone in the study of emergent spacetime, where geometric and gravitational properties in the bulk can be interpreted through quantum field-theoretic quantities on the boundary.

In recent years, discrete tensor network models have been employed to better understand and visualize aspects of the AdS/CFT correspondence, especially within the domains of quantum information theory and quantum gravity. Notably, holographic codes such as the so-called HaPPY code \cite{Pastawski_2015} have demonstrated how essential features of the duality, such as bulk reconstruction, entanglement entropy scaling, and error correction, can be captured using quantum information-theoretic structures. These tensor network models have not only deepened our conceptual understanding of holography but also forged new connections between high-energy physics and quantum information.

Despite their successes, existing tensor network models of holography remain largely ad hoc. While constructions such as the HaPPY code capture certain structural features of the AdS/CFT correspondence—like entanglement wedge reconstruction and boundary-bulk duality—they are not derived from any underlying theory of quantum gravity. Recent work has proposed to weaken the entanglement structure of known holographic codes in order to capture algebraic properties of the boundary theory in an inductive limit \cite{chemissany2025infinitetensornetworkscomplementary}, as the HaPPY code construction alone would be in violation of Lorentz symmetry.

In this work, we address this gap by embedding tensor network models into the framework of loop quantum gravity (LQG). In LQG, quantum states of geometry are represented by spin network states -- tensor networks built from representations of a group $G$ and corresponding intertwiners. \change{LQG thereby yields a direct derivation of quantum gravitational models with tensor network structure.} The group $G$ depends on spacetime dimension and metric signature considered. In the case of gravity in 3+1 dimensions, as well as in 3d Euclidean gravity relevant to the present work, the group $G$ is given by SU(2). In these cases, the spin networks naturally implement a local SU(2) Gauss constraint associated with the gauge symmetry of a Palatini-type formulation of general relativity, ultimately coming from Lorentz invariance \cite{Rovelli_1995b,Thiemann:2007pyv}. \change{We show that a consistent formulation with LQG is not directly compatible with the original proposal of perfect tensors in \cite{swingle2012constructingholographicspacetimesusing} and \cite{Pastawski_2015}, but brings us to so-called hyperinvariant tensor networks, which is a weaker isometric structure.} Before diving into the detailed analysis, we introduce the most important concepts for hyperinvariant tensor networks (Section \ref{sec:HTN}) and loop quantum gravity (Section \ref{sec:lqg}).

\change{In holographic tensor-network constructions, isometric encoding maps are closely tied to states whose boundary subsystems exhibit (approximately) maximally mixed reduced density matrices. Previous works have shown strong limitations for the compatibility of isometric tensors (or equivalently quantum states with maximally mixed marginals, cf. section \ref{sec:HTN}) that are also gauge covariant with respect to a continuous symmetry group. In \cite{Cao2024}, it was shown that stabilizer states can only probe a degenerate section of the area operator.} No-go theorems for Werner states \cite{bernards2024multiparticle, Bernards_PHD}, SU(2) intertwiners with valence four \cite{Li_2017} or local dimension two \cite{Mansuroglu_2023} can be understood in a more general context of a no-go result for covariant codes \cite{Hayden_2021} and from lower bounds for the infidelity of recovery \cite{Faist_2020}. In Section \ref{sec:nogo}, we provide an alternative proof for a no-go theorem, stating that SU(2) intertwiners cannot describe $n$-party quantum states whose two-body marginals $\rho_{1k}$ are maximally mixed for all $k = 2, ..., n$. This also excludes holographic codes with bulk degrees of freedom. 

There are, however, examples where all but one of the marginals $\rho_{1k}$ are maximally mixed (see Section \ref{sec:examples}). As instances of so-called hyperinvariant tensor networks \cite{evenbly2017hyperinvariant}, which have been studied as weaker structures for holographic codes \cite{Steinberg_2023, steinberg2024far}, we study their embedding into a spin network structure. We demonstrate that certain aspects of AdS/CFT are realized within quantum states of the gravitational field in LQG, thus providing a first-principles foundation for previously proposed models such as those in Refs.~\cite{Pastawski_2015, evenbly2017hyperinvariant, Steinberg_2023}. 

These so-called {\it Hyperinvariant, (SU(2)-) Invariant Tensors}, or HITs, thus combine properties from both approaches. As hyperinvariant tensor networks, they show an approximate duality between bipartite boundary entanglement and bulk geodesic distance, which can now be interpreted as an expectation value of the loop quantum gravitational length operator, by viewing the HIT as a quantum state of geometry. We study the geometric properties of HITs in detail in Section~\ref{sec:geo}, where we show that, on all examples that we found, the LQG length is proportional to the graph length, that is, the number of intersection points between the geodesic and the state's underlying graph, thus validating the entanglement-geometry correspondence on HITs. We further discuss the calculation of surface areas on HITs (geometrically corresponding to spatial volumes) and with this a derivation of negative curvature of the HIT network as a quantum state version of the Poincaré disc. We close on a note on two-point correlations on the discretized boundary and how decay of correlations can put further constraints on the shape of the HIT. 

Interesting results on tensor networks only containing SU(2) intertwiners have already been obtained in Ref.~\cite{Han:2016xmb} in the context of a coarse graining procedure and in Ref.~\cite{Chirco:2017vhs,Chirco:2019dlx,
Colafranceschi:2020ern,Chirco:2021chk,Colafranceschi_2022} in the context of group field theory, which is closely related to LQG. They include results of isometric embedding of bulk in boundary degrees of freedom, and Ryu-Takayanagi-type formulas. In contrast to the present work, these results employ averaging over random tensors, and, in some cases, a limit of high bond dimension. As a consequence, the relation to a specific theory of quantum gravity is not always clear. Here, we employ specific tensors (HITs), and the bond dimension is small and fixed in the examples of HITs that we consider. Consequently, we work with specific states in the LQG Hilbert space and can make precise statements about their quantum geometric properties, using established geometric operators.

The paper is organized as follows. Section~\ref{sec:HTN} introduces the necessary preliminaries, including the definitions of hyperinvariant tensors and their isometric properties. Motivated by but not necessarily dependent on any gravitational interpretation, section~\ref{sec:nogo} establishes no-go theorems that constrain the existence of certain tensor network structures, while also identifying those that are allowed. In section~\ref{sec:lqg}, we provide a concise introduction to loop quantum gravity, focusing on the aspects most relevant for the framework we consider. Section~\ref{sec:examples} shows explicit examples of quantum state built from HITs, whose geometrical properties determined by the action of the LQG length and area operators is discussed in section~\ref{sec:geo}.

\begin{figure}[tbp]
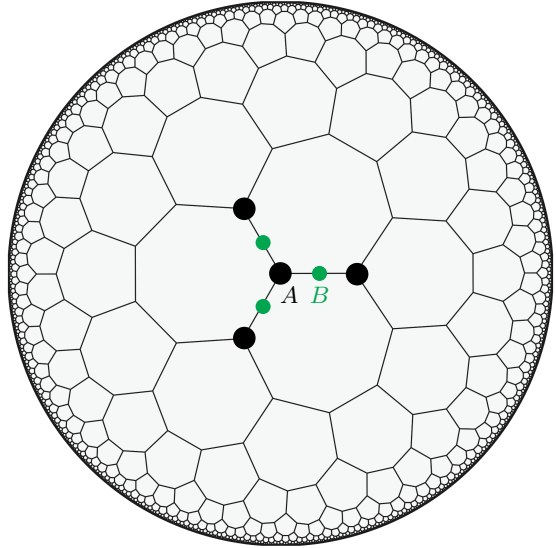

    \centering
    \begin{align*}
        \SNSdisc
    \end{align*}
    \caption{\textbf{Schematic view of the construction methods using a (7,3) tiling of the 
    Poincaré disc.} A hyperinvariant tensor network constructed from tensors $A$ on 
    all vertices and $B$ on all edges.}
    \label{fig:poincare_disk}
\end{figure}

\section{Hyperinvariant tensor networks}
\label{sec:HTN}

In this section, we review the construction and properties of hyperinvariant tensor networks introduced in Ref.~\cite{evenbly2017hyperinvariant} which have been used as toy models for AdS/CFT in previous works \cite{swingle2012constructingholographicspacetimesusing,Pastawski_2015, evenbly2017hyperinvariant, Steinberg_2023}.

\subsection{Construction of hyperinvariant tensor networks}
The tensor networks are constructed on hyperbolic tilings of the Poincar\'e disk to be compatible with the symmetry and negative curvature of the anti-de Sitter space. Hyperbolic $(p,q)$ tilings consist of $p$-gons with $q$ of them meeting at each vertex (see   \cref{fig:poincare_disk}) with $pq>2(p+q)$ ensuring a negative curvature for the quantum geometry, as we will see below. A $q$-valent (or $q$-partite) tensor $A_{i_1\ldots i_q}$ is placed on each vertex and a $2$-valent tensor $B_{i_1 i_2}$ on each edge. Indices along connected tensors are contracted resulting in a tensor with free indices at the boundary. 

These tensors have to fulfil several symmetry constraints. 
First, the tensors $A$ and $B$ have to be symmetric under cyclic permutations of the indices, i.e.\ $A_{i_q i_1\ldots i_{q-1}}=A_{i_1\ldots i_q}$ and $B_{i_1 i_2}=B_{i_2i_1}$. Second, $A$ has to be \emph{1-isometric}, i.e.\ \change{$A_{i_k;i_1\ldots i_{k-1}i_{k+1}\ldots i_q}$} defines an isometry\change{, i.e.\ preserves inner products,} from the $k$-th index to the rest of the indices for all $k$. The same holds for the tensor $B$ with $1$-isometric being equivalent to unitary in the bipartite case \cite{evenbly2017hyperinvariant,steinberg2024far}.
Finally, there is a constraint including two copies of the tensor $A$ and one of $B$. Namely,
\begin{equation}
    \sum_{\substack{i_q,j_2}} A_{i_1 i_2 \ldots i_q}B_{i_q j_2}A_{j_1 j_2 \ldots j_q}=V_{i_1j_1;i_2\ldots i_{q-1}j_3\ldots j_q}
\end{equation}
has to be an isometry from $i_1j_1$ to the complement. We summarize the requirements for hyperinvariant tensor networks with a graphical representation of these isometry conditions in the following definition.

\begin{definition}[Hyperinvariant tensor networks]\label{def:HIT}
    A hyperinvariant tensor network consists of tensors $A$ and $B$ being placed on the vertices and edges of a $(p,q)$-tesselation of the Poincar\'e disk respectively. The tensors $A$ and $B$ have to fulfil the isometry constraints 
    \begin{equation}
    \begin{aligned}
        \vcenter{\hbox{\isometryIL}}&\quad = \quad \vcenter{\hbox{\isometryIR}}\qquad \text{ and }\qquad\vcenter{\hbox{\isometrysingleB}}\quad = \quad \vcenter{\hbox{\isometryIR}}\\ &\text{ and }\qquad \vcenter{\hbox{\isometrydouble}}\quad = \quad \vcenter{\hbox{\isometryIIR}}\; 
        \label{eq:isometries_HIT}
    \end{aligned}
    \end{equation}
    \change{where each leg of a tensor represents a single index and straight lines correspond to the Kronecker-Delta $\delta_a^b$, i.e.\ they contract the indices of the tensors they are connected to.}
\end{definition}
The first holographic states based on tensor networks have been constructed in 
Ref.~\cite{Pastawski_2015}. They can be seen as a special case of hyperinvariant 
tensor networks featuring so-called \emph{perfect tensors} $A_{i_1\ldots i_n}$ 
defining isometries between every bipartition of the indices whereas the tensor 
$B_{i_1i_2}$ can be taken to be the identity.

The mentioned properties of the tensors can be translated to properties of corresponding quantum states. The tensor $T_{i_1\ldots i_n}$ being an isometry from a subset of indices $\mathcal A$ to the complement $\mathcal{A}^\complement$ translates to the marginal $\rho_{\mathcal A}:=\tr_{\mathcal{A}^\complement}\pro{\psi_T}$ of the quantum state
\begin{equation}
    \ket{\psi_T}=\sum_{i_1,\ldots,i_n} T_{i_1\ldots i_n} \ket{i_1\ldots i_n}
    \label{eq:state_tensor}
\end{equation}
being maximally mixed. Perfect tensors correspond to absolutely maximally entangled (AME) states that have been widely discussed in quantum information theory \cite{huber2017absolutely,rather_thirty-six_2022,rajchelmieldzioć2025absolutely}. These properties 
are summarized in the following definition.
\begin{definition}[$k$-isometric tensors and $k$-uniform quantum states]
    \phantom{a}\\ \vspace{-8mm}\\
    \begin{enumerate}[
    label=(\alph*)]
        \item The tensor $T_{i_1\ldots i_n}$ is $k$-isometric if it defines an isometry from any set of $k$ indices to the rest of the indices. It is called a perfect tensor if it is $\lfloor\frac n2\rfloor$-isometric.
        \item The quantum state $\ket{\psi_T}$ is $k$-uniform if all marginals $\rho_{\mathcal A}=\tr_{\mathcal{A}^\complement}\pro{\psi_T}$ on any set $\mathcal A$ of $k$ parties are maximally mixed: $\rho_{\mathcal A}\propto \ID$. It is called absolutely maximally entangled (AME) if it is $\lfloor\frac n2\rfloor$-uniform.
    \end{enumerate}
\end{definition}
\change{Note that the tensor $T$ is $k$-isometric if and only if the state $\ket{\psi_T}$ is $k$-uniform. Moreover, a $k$-uniform tensor is $l$-uniform for all $l\leq k$, because the preservation of inner products inherently extends to subsets of parties.

The} single-tensor constraint on $A$ is equivalent to the corresponding quantum state being $1$-uniform. The same holds for $B$ implying that the corresponding state is a maximally entangled bipartite state.

Finally, note that not only the tensors have counterparts on the quantum state level, but the contractions of indices to form a network, too. Contracting indices of connected tensors corresponds to entanglement swapping \change{where previously disconnected parties become entangled by an operation on the contracted indices/parties \cite{Zukowski_1993}.}

\subsection{Properties of hyperinvariant tensor networks}

\begin{figure}[t]
    \centering
    \includegraphics[width=0.95\linewidth]{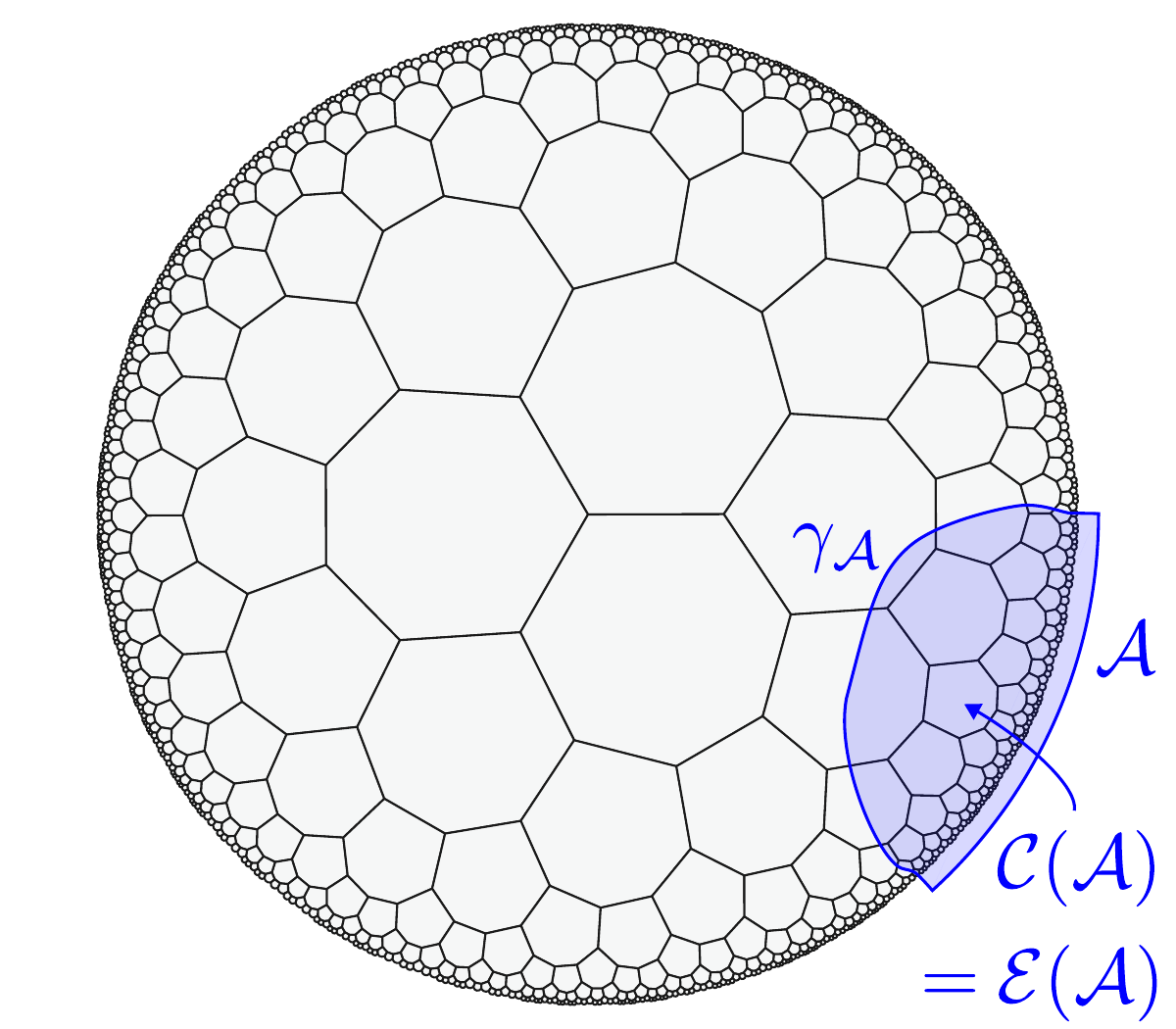}
    \caption{\textbf{Causal cone and entanglement wedge for a boundary region $\mathcal A$ in a $(7,3)$-tiling.} In this case, the causal cone and the entanglement wedge coincide.}
    \label{fig:causal_cone_entanglement_wedge}
\end{figure}

Hyperinvariant tensor networks give rise to important properties predicted by AdS/CFT. The entanglement structure on the boundary is related to the geometry of the bulk and thus the gravitational structure. For formalizing this, two important definitions are the causal cone and the entanglement wedge of a boundary region~$\mathcal A$.

\begin{definition}[Causal cone \cite{evenbly2017hyperinvariant}]
    Let $\mathcal{A}$ be a boundary region of a hyperinvariant tensor network. The reduced state on the boundary $\rho_{\mathcal{A}}= \tr_{\mathcal{A}^\complement} \pro{\psi}$ generally depends on the tensors in the bulk. A single tensor in the bulk belongs to the causal cone $\mathcal{C}(\mathcal{A})$ if the reduced state $\rho_{\mathcal{A}}$ nontrivially depends on its choice within the set of hyperinvariant tensors.
\end{definition}

\begin{definition}[Entanglement wedge \cite{Pastawski_2015}]
    The entanglement wedge $\EE(\mathcal A)$ of a boundary region $\mathcal A$ is the set of tensors in the bulk enclosed by $\mathcal A$ and the minimal surface\footnote{\emph{Minimal} is understood in the literature in terms of the number of intersections with the tensor network. We will show below that these surfaces are indeed minimal with respect of the quantum geometry of the tensor network state. If the minimal surface is not unique, we choose the one maximizing the size of the entanglement wedge.} (or geodesic) $\gamma_{\mathcal A}$ in the bulk whose end points coincides with the boundary of $\mathcal A$.
\end{definition}

The causal structure of such networks has been extensively studied in the framework of multiscale entanglement renormalization ansätze \cite{Vidal_2008,B_ny_2013, Chua_2017}, which covers hyperinvariant tensor networks. It has been shown that the causal cone and the entanglement wedge approximately coincide \cite{evenbly2017hyperinvariant}. For certain tilings, e.g., a (7,3)-tiling (cf.\ \cref{fig:causal_cone_entanglement_wedge}), the causal cone and the entanglement wedge match for any boundary region. For other tilings the difference between both vanishes with respect to their relative size in the limit of an infinitely large network.

Finally, in the special case of the tensors $A$ and $B$ being perfect, the bipartite entanglement entropy between a boundary region $\mathcal A$ and its complement $\mathcal{A}^\complement$ is related to the length of the geodesic $\gamma_{\mathcal{A}}$ via the Ryu-Takayanagi formula \cite{Ryu06,Pastawski_2015},
\begin{equation}\label{eq:rt}
    S_{\mathcal{A}}\propto L(\gamma_{\mathcal{A}}).
\end{equation}
Another prediction of AdS/CFT is the decay of two-point correlations with the distance between the two points. This has been addressed in Ref.~\cite{Bistro__2025}, where the authors use a specific realization of a hyperinvariant tensor network to realize these correlations.

\section{LQG states with boundary Qudits}
\label{sec:lqg}
In this section, we provide a brief overview of the core ideas from loop quantum gravity (LQG) that are relevant for our work. 
LQG is an approach to the quantization of gravity. It is based on a description of the gravitational field that is close to that of a gauge theory. In formulating the quantum theory, it also avoids the use of a classical space time geometry, to implement as much as possible the general covariance of the classical theory. As a result, the quantum field theory is very different than those used to describe other interactions. In the canonical approach relevant for the present work, an atomistic picture of spatial geometry emerges. 

A natural basis of quantum states, the spin network states, can be considered as special tensor network states constrained by an SU(2) symmetry implementing the gauge symmetry of gravity. The tensors describe quantum states of elementary chunks of space, and the network describes their arrangement. \change{Uncontracted indices at the boundary can be interpreted as matter degrees of freedom within LQG.} This perspective allows us to extend the hyperinvariant tensor construction introduced in Section~\ref{sec:HTN}: we define Hyperinvariant, SU(2)-Invariant Tensors (HITs) as linear combinations of spin network states. Establishing the relation between HITs and conventional spin networks is essential for computing expectation values of geometric operators in LQG, most notably the length and area operators, whose action on spin networks is well understood. For a detailed and rigorous introduction to the formalism of LQG, we refer the reader to \cite{Thiemann:2007pyv}.
In Sections \ref{sec:holonomies} and \ref{sec:spinnets} we will sketch how the Hilbert space of quantum states of geometry arises in LQG. Section \ref{sec:holonomies} explains the concept of parallel transport and holonomy, and how a part of the gravitational field can be encoded in those. This is important because the holonomies are well defined operators on the Hilbert space of spatial geometry, and they create line-like excitations that give rise to the spin network states that we identify with special tensor networks. The Hilbert space is described in Section \ref{sec:spinnets}, and Section \ref{sec:hits_lqg} connects it to the hyperinvariant tensor networks of Section~\ref{sec:HTN}. Finally, in Section \ref{sec:length_area}, we review the definition of length and area operators in LQG. These and other operators can be used to probe the quantum geometric properties of the HITs. 

\subsection{Geometry with Parallel Transport Maps}
\label{sec:holonomies}
The kinematic Hilbert space of loop quantum gravity is thoroughly studied throughout the past decades. Ultimately, it is derived from a Palatini-type formulation of general relativity in terms of a tetrad field $e$ and a connection one-form $\omega$, rather than a metric. 
The tetrad field determines the metric. The field equations are equivalent to Einstein's equations and metricity of $\omega$. In the canonical formulation of this theory, the variables become a (densitized) triad $E^a$, and a connection $C_a$ on a spatial slice $\Sigma$ of spacetime. \change{In both (2+1)d Euclidean and Lorentzian gravity}, $C_a$ is just the pullback of $\omega$ onto $\Sigma$\change{, which is two-dimensional.} Field equations take the form of constraints on $C$ and $E$ in this context, with the Gau{\ss} constraint (an SU(2) generalization of the well known Gau{\ss} law of electrodynamics) requiring gauge invariance, the diffeomorphism constraint requiring invariance under diffeomorphisms of $\Sigma$, and a Hamilton constraint as a dynamical law. Altogether, these constraints and the Hamilton equations of motion are equivalent to Einstein's equations.    

The connection $C$ induces a notion of parallel transport of vectors (in a certain fiber bundle) along curves $e$ in a spatial slice $\Sigma$. It is given by solving the parallel transport equation,

\begin{align}
    \dot v(s)^i + \sum_{k, \mu} C(e(s))_\mu {}^i {}_k \, \dot e(s)^\mu v(s)^k = 0,
    \label{eq:parallel_transport}
\end{align}
for a family of tangent vectors $v(s)$ at the point $e(s)$. The parallel transport can be encoded in so called holonomies $h_e$,  functionals in the connection $C_a$. Formally, the solution of Eq.~\eqref{eq:parallel_transport} can be written as $v(1) = h_{e}[v(0)]$ with the holonomy being a path-ordered exponential $h_e = \mathcal{P}\exp\{ - \int ds \sum_\mu C(e(s))_\mu \dot e(s)^\mu \}$ of the matrix-valued one-form $C_\mu$. In the case considered here, $C_\mu$ is an $\mathfrak{su}(2)$-valued one-form\footnote{Recall that the adjoint representation of the Lie algebra of SU(2) is isomorphic to the three-dimensional representation of $\mathfrak{so}$(3) via $X^a {}_b \to -i (*X)^a \sigma_a$ with the Pauli matrices $\sigma_a$.} and can hence be considered anti-symmetric, and generating three-dimensional rotations. The action on a vector of dimension $2j+1$ can be understood as a spin $j$ representation of SU(2), $\pi_j(h_e[A])$, acting as a map between the two tangent spaces:
\begin{align}
    \vcenter{\hbox{\holonomy}}
\end{align}
The $2j+1$-dimensional vector $v(0)$ is mapped from point $e(0)$ to $e(1)$, given a connection $C$. The limiting case where $\pi_j(h_e[C]) = \mathds{1}$ for all paths $e$ corresponds to trivial parallel transport, i.e., a flat connection \cite{Rovelli_1995,Thiemann:2007pyv,Lewandowski_2006}.
\begin{definition}[Holonomy]
    Let $e: [0,1] \to \mathcal{M}$ be a smooth curve on a manifold $\mathcal{M}$. A holonomy $h_e[C]$ is a solution to the parallel transport equation, Eq.~\eqref{eq:parallel_transport}. A spin-$j$ holonomy is a $2j+1$-dimensional irreducible representation of $\SUT$.
\end{definition}
In the canonical formulation of LQG, matrix elements of holonomies are well defined operators. They act as creation operators and produce gravitational excitations from a vacuum state devoid of geometry.  

Holonomies (both the classical maps and their quantum counterparts) transform non-trivially under gauge transformations $g:\Sigma\rightarrow \text{SU(2)}, x\mapsto g(x)$,
\begin{equation}
    h_e \mapsto g^{-1}(e(1))\; h_e\; g(e(0)). 
\end{equation}
Holonomies $\pi_j(h_e)$ taken in the spin $j$ irrep are transforming in an analogous fashion, with $\pi_j(g(x))$ replacing $g(x)$. 

In order to connect holonomies at their respective start- and endpoints in a covariant way, the vertices are decorated by a tensor describing a homomorphism of representations, a so called intertwiner. ``Covariant'' here means that the resulting structure again only transforms at the (unconnected) endpoints, and not at the vertices. Covariance also guarantees that a vector that is transported across a vertex can be mapped to (potentially many delocalized) vectors without losing information. \change{Consider the following example of a spin 1 vector $\psi(p_1)_k$ that is mapped by holonomies $h_{e_i}, i=1,2,3$ meeting at an intertwiner $I^k_{\mu_2 \mu_3}$. The spin $\frac{1}{2}$ vectors $\psi(p_{2/3})_{\mu_{2/3}}$ can be defined via the parallel transport
\begin{align}
    \vcenter{\hbox{\intertwiner}},
\end{align}
which corresponds to the algebraic formulation
\begin{align}
    \psi_{\mu_2} \otimes \psi_{\mu_3} = \psi_k \tensor{h}{_{e_1}^{k}_l} I^l_{\nu_2 \nu_3} \tensor{h}{_{e_2}^{\nu_2}_{\mu_2}} \tensor{h}{_{e_3}^{\nu_3}_{\mu_3}}.
\end{align}
}

\subsection{Spin Network States}
\label{sec:spinnets}
Physical states in LQG must, among other things, be gauge invariant. Therefore intertwiners are useful in their construction, as we will describe below. To be precise, we consider intertwiners according to the following definition.
\begin{definition}[Intertwiner]
\label{def:inv_int}
    Let $\ket\psi\in\bigotimes_{i=1}^n \HH_{d_i}$ be a $n$-partite quantum state and $R(g)=\bigotimes_{i=1}^n R^{(i)}(g)$ a representation of the group $G$ where each $R^{(i)}: G\to \operatorname{GL}(\CC^{d_i})$ is a $d_i$-dimensional representation of $G$. The state $\ket\psi$ is called \textbf{$G$-invariant} with respect to the representations $R^{(i)}$ if
    \begin{equation}
        R(g)\pro{\psi}R(g)^\dagger=\pro{\psi}
        \label{eq:symmetry}
    \end{equation}
    for all $g\in G$. The tensor $T$ defining $\ket\psi$ via Eq.~\eqref{eq:state_tensor} is called \textbf{$G$-intertwiner}. Interpreted as a map from one partition to its complement, the tensor $T: \bigotimes_{i\in \mathcal{I}} \HH_{d_i} \to \bigotimes_{i\in \mathcal{I}^\complement} \HH_{d_i}$ is called \textbf{covariant map}.
\end{definition}
Although in this section, we only need the definition for spin $j_i$ representations of $G=\SUT$ of dimension $d_i = 2j_i+1$, we later also consider U(1) representations and thus keep Definition \ref{def:inv_int} general. Note that states with invariance properties as in Eq.~(\ref{eq:symmetry}) have been widely discussed in quantum information theory \cite{Werner89, Eggeling_2001, Cabello_2003}. \change{For the definition of spin network states, we assume a slightly stricter condition, $R(g) \ket{\psi} = \ket{\psi}$, to avoid bookkeeping with global phases.}

The kinematic Hilbert space $\mathcal{H}_{\rm kin}$ of loop quantum gravity is finally constructed by holonomies and intertwiners contracted in an SU(2)-invariant fashion. For a given graph $\Gamma$, the graph Hilbert space 
\begin{equation}
    \mathcal{H}_\Gamma \cong L^2(SU(2)^{|E(\Gamma)|}, \text{d}g^{|E(\Gamma)|})
\end{equation}
is spanned by states created by holonomies along the edges $e \in E(\Gamma)$. 
$\text{d} g$ denotes the invariant measure on SU(2) \cite{ashtekar1993representationtheoryanalyticholonomy}. Gauge transformations act unitarily on this space, and the invariant states $\text{Inv}(\mathcal{H}_\Gamma)$ are the physically relevant ones. They form subspaces of $\mathcal{H}_{\rm kin}$, and in fact 
$\mathcal{H}_{\rm kin} = \bigoplus_\Gamma {\rm Inv}(\mathcal{H}_\Gamma)/\sim$, where $\sim$ indicates a certain equivalence relation that has to be taken into account when graphs partially overlap or are nested inside each other.\footnote{The precise statement is that $\mathcal{H}_{\rm kin}$ can be constructed as a certain direct limit. For details see Ref.~\cite{Ashtekar:1994mh}.} In the following, it is enough to restrict to particular graphs $\Gamma$. 

There is a practical orthogonal basis for Inv $\mathcal{H}_\Gamma$, consisting of specific tensor network states in which the holonomies are in irreducible representations and are connected by intertwiners.
\begin{definition}[Spin Network States]
    A state $\Psi$ in $\operatorname{Inv}(\mathcal{H}_\Gamma)$ is called a \textbf{spin network state} if it is of the form     
    \begin{equation}
        \Psi = \Tr \left[\pi_{j_1}(h_{e_1})\ldots \pi_{j_{|E(\Gamma)|}}(h_{e_{|E(\Gamma)|}})\; I_{1} \ldots I_{_{|V(\Gamma)|}}  \right], 
    \end{equation}
    where the $I$ are intertwiners for the $\pi_j$ meeting at a vertex, and the trace is a shorthand for the gauge invariant contractions.
\end{definition}
Matter can be included in spin network states by coupling it directly to an end of a holonomy $\pi_{j}(h_e)$, or to several holonomies meeting in a vertex via a suitable intertwiner. An integer spin $j$ can only be coupled to bosons and half-integer spin to fermions \cite{Mansuroglu_2021_kinematics, Mansuroglu_2021_fermions}. Matter degrees of freedom can be depicted by free indices, i.e. indices that are not coupled to an intertwiner, see Fig.~\ref{fig:SNS}.
\begin{figure}
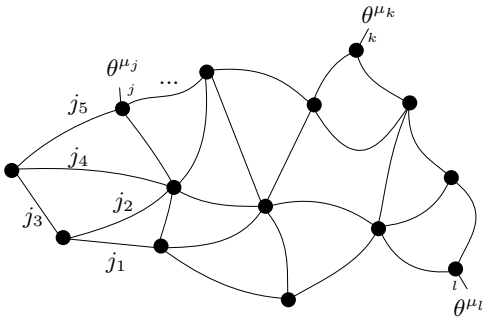

    \centering
    $\vcenter{\hbox{\SNS}}$
    \caption{\textbf{Pictorial representation of a spin network state with matter degrees of freedom.} The spin quantum numbers, $j_i$, determine the SU(2) representation of the holonomy along the corresponding edge. The vertices depict intertwiner and the matter fields, $\theta^\mu$, couple free indices in an SU(2)-invariant fashion.} 
    \label{fig:SNS}
\end{figure}

Next to the SU(2) Gauß constraint, \change{which implements the internal rotational symmetry of the metric}, spin network states are subject to two other constraints, one that accounts for diffeomorphism symmetry and one that identifies equivalent foliations into spacelike hypersurfaces $\mathcal{M}$, for technical details and derivations, see Refs.~\cite{Thiemann_1998, Thiemann:2007pyv}. 

\subsection{Hyperinvariant Loop Quantum Gravitational States}
\label{sec:hits_lqg}
In this paper, we are interested in $(p,q)$ tilings of the Poincaré disc as a graph description for spin network states. The Hilbert space consists of a boundary and a bulk subsystem $\mathcal{H}_{\rm bulk} \otimes \mathcal{H}_{\rm boundary}$. The boundary Hilbert space $\mathcal{H}_{\rm boundary}$ consists of a number $N$ of qudits with local dimension $d$ that relates to a discretized CFT. The bulk Hilbert space, on the other hand, consists of all degrees of freedom from loop quantum gravity, that are holonomies, intertwiners and possibly qudits living on bulk vertices. 

Given a $(p,q)$ tiling as shown in Fig.~\ref{fig:poincare_disk}, we want to define a class of states which are built from hyperinvariant tensors \cite{evenbly2017hyperinvariant, Steinberg_2023} and lie within the kinematic Hilbert space described above. To this end, we consider $\SUT$ intertwiners $A$ with valence $q$, which have an additional invariance under cyclic permutation of indices, and tensors $B$ with valence 2 that are additionally solutions to the isometry relations displayed in \cref{def:HIT}. These are the defining relations for hyperinvariant tensor network states as introduced in Refs.~\cite{evenbly2017hyperinvariant, Steinberg_2023} and summarized in Section \ref{sec:HTN}, just that the edges are described by holonomies carrying gravitational excitations weighted by spin labels $j$ and tensors that are also $\SUT$ invariant. 

The hyperinvariance condition is extended by the decoration of the edges by holonomies. Absorbing them into $B$ recovers Eq.~\eqref{eq:isometries_HIT} \change{using
\begin{align}
    \vcenter{\hbox{\pushingB}} =: \vcenter{\hbox{\pushingBIII}},
\end{align}
and dropping the tilde again to avoid clutter. By unitarity the holonomies including $B$ cancel on all indices that are being traced out and in the 2-isometry condition only one pair with $B$ and $\bar B$ survives.}

For higher valence $q$, more \change{than one index per tensor $A$ is pointing inward and an entanglement wedge boundary may intersect multiple legs of the same tensor (depending on the tiling).} Therefore additional isometries including three or more contractions of $A$ might be necessary to relate entanglement properties of the boundary state to geometric properties of the bulk that has been found in \cite{Ling_2019, Steinberg_2023}. We show later that $\SUT$ symmetry prohibits 2-uniformity, \change{which makes hyperinvariant tensor networks candidates for loop quantum gravitational states with an bulk-boundary correspondence, because they are isometries only on neighboring indices and not all permutations of $k$ indices as $k$-uniformity would require. We will see that this enables hyperinvariant tensor networks with valence $q \geq 4$, which are still required to be 2-isometric for neighboring indices, but not 2-uniform, to be viable candidates for LQG.}

More interesting solutions arise when allowing reducible $\SUT$ representations on the edges of the spin network state. On this ``fine-structure'', we impose $\SUT$ invariance and the isometry relations on the ``coarse structure''. A more accurate graphical depiction would hence be
\begin{align}
    &\hspace{0.5cm} \vcenter{\hbox{\isometryIIIL}} \quad = \quad \vcenter{\hbox{\isometryIIIR}} \qquad \text{and} \nonumber \\
    &\vcenter{\hbox{\isometryIVL}} \quad = \quad \vcenter{\hbox{\isometryIVR}},
    \label{eq:isometries_reducible}
\end{align}
where the thick lines have been decomposed into $k$ irreducible components (we depicted $k=3$ for the sake of simplicity). Since the holonomies are $\SUT$ representations and the intertwiners are SU(2)-invariant, the holonomies can be pushed to the input \change{and output indices using SU(2) invariance
\begin{align}
    \vcenter{\hbox{\pushing}} = \vcenter{\hbox{\pushingII}}.
\end{align}
With unitarity of $U$ the isometry relations ultimately become a property of the tensors $A$ and $B$.}

\subsection{Length and Area Operators}
\label{sec:length_area}
Our goal is to show a relation between bipartite entanglement on the boundary and minimal path length in the bulk. While previous literature simply considers a graph length, i.e. the number of intersection points between a path and a graph \cite{Pastawski_2015, Steinberg_2023}, we compare against expectation values of the loop quantum gravitational length operator, which can be derived by canonical quantization of the integrated line element in two dimensions \cite{Rovelli_1995}. Spin network states are the eigenstates of the length operator $L_\gamma$ of a curve $\gamma$ to the eigenvalue $\sqrt{j (j+1)}$ with $j$ being the spin quantum number of the edge intersecting $\gamma$:
\begin{align}
    \vcenter{\hbox{\LQGarea}},
\end{align}
where $j$ denotes the state projected onto the spin $j$ subspace. We can thus define the length operator as a sum of projectors.
\begin{definition}[Length Operator]
    Let $\gamma \hookrightarrow \mathcal{M}$ be a curve embedded in a manifold $\mathcal{M}$. The loop quantum gravitational length operator acts on the subsystems of holonomies $\bigotimes_{e \cap \gamma \neq \emptyset} \mathcal{H}_e$ which intersect $\gamma$ as
    \begin{align}
        L_\gamma = \int_\gamma dx \sum_j \sqrt{j(j+1)} P_x^{j},
    \label{eq:length_op}
    \end{align}
    with the operator $P_x^j$ projecting onto the spin $j$ subspace of holonomies at the point $x$.
\end{definition}
For the sake of simplicity, we have set the Planck length to one in Eq.~\eqref{eq:length_op}. It is straightforward to calculate the curve length connecting two points at the boundary of a quantum gravitational state by first decomposing it in the basis of spin network states, where the expectation value of $L_\gamma$ just becomes an average over the eigenvalues of spin network states. As we will see later, a decomposition into spin network states is not necessary to calculate the expectation value of the length operator in some simple cases. 

This is different for the 2-area, or 2D-volume operator, of a surface $\mathcal{S}$ which can be written as a sum of local operators $S_v$ acting on the Hilbert space of intertwiners at the vertices $v$ contained in $\mathcal{S}$ only \cite{Freidel_2003, Thiemann_1998_QSD, Long_2020}. Three-valent intertwiners are eigenstates of $S_v$ to an eigenvalue $a_{jkl}$ that depends only on the spin quantum numbers $j,k,l$ of the three indices of the intertwiner:
\begin{align}
    \vcenter{\hbox{\LQGvolume}},
\end{align}
\change{using the definition of the area operator from \cite{Thiemann_1998_QSD}, that is}
\begin{definition}[Area Operator]
    Let $\mathcal{S} \hookrightarrow \mathcal{M}$ be a surface embedded in a manifold of dimension 2. The loop quantum gravitational area operator acts on the local Hilbert spaces of intertwiners at vertices $v \in V(\Gamma)$ inside the surface $\mathcal{S}$
    \begin{align}
        &S_{\mathcal{S}} = \sum_{v \in \Gamma \cap \mathcal{S}} S_v \nonumber \\
        &= \sum_{v \in \Gamma \cap \mathcal{S}} \sqrt{ \left( \sum_{e_1 \cap e_2 = v} \sgn(\dot e_1(0), \dot e_2(0))  \vec X(e_1) \times \vec X(e_2) \right)^2},
    \label{eq:area_op}
    \end{align}
    by creating the unique three-valent intertwiners $X^i(e_i)$ with spins $(j_i, j_i, 1)$ on pairs of curves with linear independent tangent vectors $\dot e_1(0), \dot e_2(0)$.
\end{definition}
Three-valent intertwiners, as well as two-valent vertices with linear independent tangent vectors are eigenstates of $S_{v}$ as shown in \cite{Thiemann_1998_QSD}. A direct action on arbitrary SU(2) intertwiners can be understood for the squared operator $S_{v}^2$. The operator $X^i(e_1)$ creates an intertwiner with a free spin 1 index at the point $v$:
\begin{align}
    X^i(e_1)\left( \vcenter{\hbox{\threeVal}} \right) = \vcenter{\hbox{\grasping}}
    \label{eq:grasping}
\end{align}
For readability, we stick to three-valent intertwiners although Eq.~\eqref{eq:grasping} holds for any intertwiner. We also inserted a trivial spin $j$ holonomy between the newly created intertwiner and the old intertwiner, depicting trivial tensor contraction. Inserting the action of $\vec X(e)$ into $S_v^2$ thus yields a straightforward way to calculate its action. In Appendix \ref{app:grasp}, we show its action on the three-valent intertwiner example. However, iterating over the combinations of pairing pairs of edges makes the direct computation of the action of $S_v^2$ tedious. For an intertwiner of valence $q$, a number of $\mathcal{O}(q^4)$ many terms have to be accounted for. We stick to a calculation in a diagonal basis in section \ref{sec:geo}, where we calculate geometric properties of a HIT example enabled through the loop quantum gravitational derivations of length and area operator.

\section{Non-existence of  U(1)-invariant two-uniform states}
\label{sec:nogo}
Before we discuss examples of HITs, let us show clear limitations 
of what structures can or cannot be combined. We show that for a 
$\SUT$-invariant state, at least one of the two-body marginals cannot 
be maximally mixed, thus excluding $\SUT$-invariant AME states. 
Lifting the AME condition to specific isometries as in 
Eq.~\eqref{eq:isometries_HIT} allows the existence of HITs. 
\change{However, adding a single logical qudit onto the vertices requires stronger isometry conditions, which are again incompatible with $\SUT$-invariance and lead to the non-existence of $\SUT$-invariant holographic codes.}

\subsection{No Invariant AME States}
In this section, we show that $\SUT$-invariant AME states (equivalently invariant perfect tensors \cite{Li_2017}), which are necessary for a 
LQG equivalent to the toy models for the AdS/CFT correspondence in \cite{Pastawski_2015}, do not exist. We start by excluding the compatibility of weaker properties of the state, i.e.\ $\UO$-invariance and restrictions on a subset of the marginals.

From here on, we assume that the adjoint of the representation $R(g)$ of $\UO$ acts nontrivially, i.e.\ $R^{(k)}(g)[\cdot]R^{(k)}(g)^\dagger\neq \operatorname{id}(\cdot)$ for at least one party $k$. Otherwise, every state would be $\UO$-invariant with respect to $R(g)$.
\begin{lemma}\label{lem:inv_mixed_nogo}
    Let $\ket\psi\in\bigotimes_{i=1}^n \HH_{d_i}$ ($n\geq 4$) be 
    a $\UO$-invariant state with respect to the representations 
    $R^{(i)}$.
    Then, for every party $k \in \{ 1, ..., n \}$ with $R^{(k)}(g)[\cdot]R^{(k)}(g)^\dagger\neq \operatorname{id}(\cdot)$ there is a party $l\neq k$ such that the two-party marginal $\rho_{\{ k,l \}}$ fulfils 
    $\rho_{\{ k,l \}}\neq\frac{\ID}{d_k}\otimes \rho_l$. This implies 
    that $\rho_{\{ k,l \}}$ cannot be a maximally mixed two-particle state.
\end{lemma}
\begin{proof}
    The generator of the group action can be expanded as
    \begin{equation}\label{eq:U1_gen}
        \hat N = \sum_i N_i=\sum_{i}\left(\alpha_i\ID^{(i)} + \sum_{a} \beta_{i,a} T_a^{(i)}\right),
    \end{equation}   
    where $T_a^{(i)}$ are chosen traceless and act nontrivially only on the $i^\text{th}$ party. Its adjoint representation is nontrivial if and only if there is a party $k$ such that $\beta_{k,a}\neq 0$ for at least one $a$. We can restrict ourselves to the case of $\alpha_i=0$ for all $i\in\{1,\ldots,n\}$ as the identity only introduces global phases which do not affect the density matrix:
    \begin{equation}\label{eq:U1_gen_eff}
        \hat N =\sum_{i,a}\beta_{i,a} T_a^{(i)}.
    \end{equation}   
    We prove the Lemma by contradiction. Let us assume that there is a pure state $\rho$ being $\UO$-invariant with respect to the representations 
    $R^{(i)}$ and $\rho_{\{ k,l \}} = \frac{\ID}{d_k}\otimes \rho_l$ for all $l\neq k$ with $k$ being a party with a nontrivial generator, i.e.\ $\beta_{k,a}\neq 0$ for at least one $a$.
    
    We decompose the state $\rho=\pro{\psi}$ in strings of generators of the local (traceless) generators $T^{(i)}_a$ of the Lie algebra $\mathfrak{su}(d_i)$:
    \begin{align}
        \rho = \frac{1}{\prod_{i=1}^n d_i} \left( \mathds{1} + \sum_{j=1}^n P_j \right),
        \label{eq:state_decomp}
    \end{align}
    where
    \begin{align}
        P_j = \sum_{i_1 < ... < i_j} \sum_{a_1 = 1}^{d_{i_1}^2-1}\cdots\sum_{a_j = 1}^{d_{i_j}^2-1} r_{a_1, ..., a_j}^{i_1 ... i_j} \, T_{a_1}^{(i_1)} \otimes ... \otimes T_{a_j}^{(i_j)}
    \end{align}
    with real coefficients $r_{a_1, ..., a_j}^{i_1 ... i_j}$. The terms $P_j$ gather all terms with nontrivial support on exactly $j$ parties while the non-appearing subsystems feature an identity operator. We call $j$ the \emph{weight} of these terms. Note that similar decompositions have been
    studied in the context of so-called sector lengths of quantum states \cite{aschauer_2004_local,wyderka_2020_characterizing,serranoensástiga2025multiqubit}.

    The condition
    \begin{equation}
        \rho_{\{k,l\}}=\frac{\ID}{d_k}\otimes\rho_l
    \end{equation}
    on the marginals translates to
    \begin{equation}
        \rho_{\{k,l\}}=\frac{1}{d_k d_l} \mathds{1} + \frac{1}{d_k d_l} \sum_{j \leq 2} P_j\Big\vert_{\{k,l\}} = \frac{\ID}{d_k}\otimes\rho_l
    \end{equation}
    by applying the decomposition of the state as in \cref{eq:state_decomp} for all subsystems $l\neq k$. $P_j\Big\vert_{A}$ denotes the restriction of $P_j$ on $A$, that is, it contains only terms without support in $A^\complement$. From the orthogonality 
    of the $T_{a_1}^{(i_1)} \otimes \ldots \otimes T_{a_j}^{(i_j)}$ we deduce $r^k_a = 0$ and $r_{ab}^{kl}=0$ for all $a,b$ and $l\neq k$.

    The $\UO$-invariance condition $R(g)\rho R(g)^\dagger=\rho$ for a pure state $\rho=\ket{\psi}\bra{\psi}$ is equivalent to the infinitesimal relation $\hat N\ket\psi= c \ket\psi$ with $c\in\mathds{R}$. Consequently, we can express the invariance as
    \begin{equation}\label{eq:inv}
        \hat N \rho=c\rho.
    \end{equation}
    We exploit the decomposition of $\rho$ from \cref{eq:state_decomp} and examine the marginal at party $k$ on the LHS of \cref{eq:inv}:
    \begin{align}
        \Tr_{\{k\}^\complement}(\hat N \rho) &= \sum_a \frac{\beta_{k,a}}{d_k}T_a^{(k)}+\change{\frac{1}{\prod_{i=1}^n d_i}}\sum_{j=1}^{n}\Tr_{\{k\}^\complement}\left(\hat N P_j\right)
        \label{eq:contradiction}
    \end{align}
    As $r^k_a = 0$ and $r_{ab}^{kl}=0$ for all $a,b$ and $l\neq k$ the contribution of $\Tr_{\{k\}^\complement}\left( \hat N P_j\right)$ can only be proportional to the identity stemming from $\hat N P_1$. The contribution from $\sum_a \frac{\beta_{k,a}}{d_k}T_a^{(k)}$ is nonzero, because we required $\beta_{k,a}\neq 0$ for at least one $a$. Nevertheless, the marginal at party $k$ computed with the RHS of \cref{eq:inv} is proportional to the identity. This is a contradiction, because the generators $T_a^{(k)}$ are linearly independent from the identity matrix. Thus, there cannot be a state with the required properties.
\end{proof}
This Lemma directly excludes the compatibility of stronger properties of the marginals with \change{(special) unitary symmetries as all of the unitary groups have a $\UO$ subgroup}.
\begin{corollary}\label{cor:nogo}
    Let $\ket\psi\in\otimes_{i=1}^n \HH_{d_i}$ be a $n$-partite quantum state.
    \begin{enumerate}[label=(\alph*)]
        \item The state $\ket\psi$ cannot be 2-uniform and U(1)-invariant.
        \item The state $\ket\psi$ cannot be 2-uniform and SU(2)-invariant.
        \item The state $\ket\psi$ cannot be 2-uniform and U(d)-invariant with respect to the fundamental representation for $d_i=d$ for all $i\in\{1,\ldots,n\}$.
        \item The state $\ket\psi$ cannot be AME and SU(2)-invariant with respect to the $d_i$-dimensional irreducible representations if $n\geq 4$.
    \end{enumerate}
\end{corollary}
Certain results of the corollary have been discovered before. Part (c) has been shown in Refs.~\cite{Bernards_PHD, bernards2024multiparticle} and part (d) in Ref.~\cite{Hayden_2021} (for qubits also in Ref.~\cite{Mansuroglu_2023}). The techniques used in Ref.~\cite{Hayden_2021} can be applied to proof of \cref{lem:inv_mixed_nogo}, too. See \cref{app:err_corr_proof}.

Weaker properties of the marginals can, however, be compatible with unitary symmetries. In \cref{app:n_marginals}, we determine the maximal number of equal bipartitions where the marginals are maximally mixed assuming a $\SUT$ symmetry of the global state. These states can be used for the tensor networks as later shown in \cref{sec:examples}.

\subsection{No invariant Evenbly codes}
\change{In this subsection, we show the non-existence of $\SUT$-invariant holographic codes in the sense of Ref.~\cite{evenbly2017hyperinvariant,steinberg2024far}. Holographic codes feature an additional open logical leg $i_0$ for each tensor $A$ which should be reconstructable from the boundary. This leads to a new set of isometry conditions the tensors in a hyperinvariant tensor network have to fulfil \cite{steinberg2024far}. They include invariance under cyclic permutations of the physical legs $i_1,\dots,i_n$,
\begin{equation}
    A_{i_0;i_1\dots i_n}=A_{i_0;i_n i_1\dots i_{n-1}}\ ,
\end{equation}
and the tensor $A$ has to be an isometry from the logical leg plus one physical leg to the other legs}, or equivalently: the marginal of the corresponding quantum state has to be maximally mixed on the combined subsystem of the logical leg and a physical leg. Combining these two properties we get that
\begin{equation}
    T_{i_0;i_j|i_1\dots i_{j-1}i_{j+1}\dots i_n}
\end{equation}
has to be an isometry from $i_0;i_j$ to the complement for every $j\in[n]$. This translates to $\rho_{\{i_0,i_j\}}\propto\ID$ for all $j\in\{1,\dots,n\}$ where $\rho$ is the quantum state corresponding to $T_{i_0;i_1\dots i_n}$. However, this is not compatible with $\SUT$-invariance due to \cref{cor:nogo}~(b).
\begin{lemma}\label{lem:nocode}
    There is no $\SUT$-invariant holographic code.
\end{lemma}
This result can also be obtained by considering the network as a whole. As for a single tensor, error-correcting codes can be excluded in the presence of $\SUT$-invariance \cite{Hayden_2021}. \change{Note that Lemma \ref{lem:nocode} does not prohibit a collection of HITs satisfying certain isometries, but that the collection cannot be indexed in a gauge covariant fashion. In particular, this means that holographic maps from bulk to boundary might be realizable for bulk degrees of freedom that transform trivially under the gauge group. Codes from random tensors \cite{Hayden_2016, Li_2018} are not excluded by Lemma \ref{lem:nocode}, since they only require approximate isometry relations. Invariant tensors are generically close to perfect for large local dimensions \cite{Li_2017} and models related to LQG that use random tensors have been considered in \cite{Chirco:2017vhs, Chirco:2019dlx}.}

\section{Existence of SU(2)-invariant Hyperinvariant tensor network states}
\label{sec:examples}

Considering the restrictions proven in the previous section, we present examples of allowed structures that combine hyperinvariance and $\SUT$ invariance. Given a hyperbolic $(p,q)$ tesselation of the Poincaré disc with $(p-2)(q-2) > 4$, we can define several classes of examples of HITs. After discussing different possibilities to distribute Bell pairs in a way symmetric under cyclic permutations, we show examples where entanglement is distributed among multiple parties.

The simplest examples that trivially satisfy $\SUT$ invariance are maximally entangled (Bell) pairs between pairs of indices. In order to also be partial isometries, i.e. satisfy Eqs.~\eqref{eq:isometries_HIT} or \eqref{eq:isometries_reducible}, respectively, as well as cyclic permutation invariance, additional structure has to be applied. We start with a trivial example.
\begin{example}[Star Shaped Distribution]
    \label{ex:star}
    Consider an even valence $q \in 2 \mathds{N}$ and define $A$ as a tensor product of pairs of opposite indices, for instance for $q=8$
    \begin{align}
        A = \vcenter{\hbox{\starA}}.
    \end{align}
    The tensor $B$ is an arbitrary intertwiner on the reducible space (i.e. any permutation of the $k$ subspaces).
\end{example}
The \cref{ex:star} \change{is a \emph{planar maximally entangled} state} \cite{Doroudiani2020,wang2021planar} featuring maximally mixed marginals for any subset of adjacent parties (with respect to a given topology). This property allows for the application of the same techniques used in
Ref.~\cite{Pastawski_2015} to prove the Ryu-Takayanagi formula \cref{eq:rt}.
\begin{example}[Left/Right Distribution]
    \label{ex:LR}
    For arbitrary valence $q$, but fixed number of irreps $k=2$, we define $A$ as connecting neighboring parties with Bell pairs, for instance for $q=3$
    \begin{align}
        A = \vcenter{\hbox{\LRA}}, \qquad B = \vcenter{\hbox{\LRBI}}
        \label{eq:LR_state3}
    \end{align}
    and for $q=4$,
    \begin{align}
        A = \vcenter{\hbox{\LRAII}}, \qquad B = \vcenter{\hbox{\LRBI}}.
        \label{eq:LR_state4}
    \end{align}
    If $B$ is chosen as the swap operator between subspace 1 and subspace 2, one can straightforwardly show Eqs.~\eqref{eq:isometries_reducible}.
\end{example}
One can straightforwardly show that the isometry relations from Eq.~\eqref{eq:isometries_reducible} are satisfied. The 1-isometry relations are satisfied by construction. Inserting a $q$-valent tensor of the form of Example \ref{ex:LR} into the 2-isometry relation yields an equation of the form
\begin{align}
    \vcenter{\hbox{\isometryExample}} \propto \mathds{1},
\end{align}
with a proportionality constant that is determined by normalization.
\begin{example}[$l$ moves left/right]
    \label{ex:LR_general}
    For arbitrary $q$ and $2 \leq k \leq q - 1$, it is possible to distribute Bell pairs in a rotational invariant fashion by connecting subspaces with the corresponding subspace of the $l^\text{th}$ neighbor. If $k$ is large enough, this can be done simultaneously within the same tensor with $l \in \{1, ..., \lfloor\frac{q}{2} \rfloor\}$. Here are examples for $q=5$ and $k=2$
    \begin{align}
        A = \vcenter{\hbox{\LRAIII}}, \qquad B = \vcenter{\hbox{\LRBIII}} 
    \end{align}
    and for $k=4$
    \begin{align}
        A = \vcenter{\hbox{\LRAIV}}, \qquad B = \vcenter{\hbox{\LRBIV}}.
    \end{align}
    If neighboring parties share a Bell pair, we need to choose $B$ such that it swaps the two subspaces involved. This is a necessary condition to fulfil Eqs.~\eqref{eq:isometries_reducible}.
\end{example}
In fact, Examples \ref{ex:star} and \ref{ex:LR} are special cases of Example \ref{ex:LR_general}. Another construction allows us to take the tensor product of HITs.
\begin{example}[Tensor Products]
    Let $(A, B)$ and $(A', B')$ be HITs. The tensor product $(A \otimes A', B \otimes B')$ is also a HIT. One example for $q=4$ and $k=3$ is
    \begin{align}
        A = \vcenter{\hbox{\LRAV}}, \qquad B = \vcenter{\hbox{\LRBV}}.
    \end{align}
\end{example}
Although the examples listed here ultimately can be described by a bipartite structure on the quantum gravitational level, they show multipartite entanglement by the coarse graining of informational parties, see Appendix~\ref{app:GM} for a discussion on the geometric measure.

Linear combinations of our examples remain SU(2)-invariant, however they are no longer isometries, in general. This does not exclude the existence of HITs being superpositions of Bell pairs \change{(like multi-partite entangled states discussed in \cite{Akers_2020, Mori_2026, li2025tripartitehaarrandomstate})}, but we have not found any such example.

\section{Hyperbolic Geometry from Loop Quantum Gravity}
\label{sec:geo}
The geodesic length in AdS/CFT models is often calculated as the graph length of the tensor network state, i.e. the number of intersections between geodesic and graph. Intuitively, this corresponds to treating every plaquette that surrounds a vertex as an equal quantum of spacetime, making its boundary a measure for length. From loop quantum gravity, we have an alternative way of deriving the length of a curve as an expectation value of an operator from fundamental principles of geometry. In this section, we show that for a specific class of HITs, the expectation value of the length operator indeed coincides with the graph length, which enables the use of prior results on hyperinvariant tensor networks \cite{evenbly2017hyperinvariant, Steinberg_2023, Bistro__2025}. Additionally, we observe that the area of a surface $\mathcal{S}$ on the Poincaré disc scales with the number of vertices within $\mathcal{S}$. The loop quantum gravitational area operator enables us to also compute the proportionality factor for a given tensor~$A$.

\subsection{Geodesic Length}
\label{sec:geodesic}
Using the length operator from Eq.~\eqref{eq:length_op}, it is straightforward to calculate the minimal curve length connecting two points at the boundary of a state like the ones in Eq.~\eqref{eq:LR_state3} and \eqref{eq:LR_state4} by evaluating the expectation value as an average over the decomposition into spin network states, which are eigenstates of the length operator. One can always decompose a state into spin network states, since they form a basis of the SU(2)-invariant states \cite{Rovelli_1995b, Thiemann:2007pyv, Rumer1932}. But even without a spin network decomposition, we can calculate \change{the expectation value of the length operator at an intersection point $x \in \gamma \cap \Gamma$}
\begin{align}
    &\bra{ A } L_\gamma \ket{ A } = \sum_j \sqrt{j(j+1)} \braket{ A | P_x^j | A } \nonumber \\
    &= \sum_j \sqrt{j(j+1)} \sum_{\ket{\phi} \in \mathcal{P}(P_x^j)} \lvert \braket{ A | \phi } \rvert^2 =: \sum_j \ell_j,
    \label{eq:length_contribution_A}
\end{align}
with $\mathcal{P}(P_x^j)$ denoting the +1 eigenspace of the projector $P_x^j$. The state $\ket{A}$ is a placeholder for any single vertex tensor consisting of a number $m$ of Bell pairs. With the notation $\ket{A}$, we emphasize that the tensor $A$ (together with holonomies contracted to the open indices of $A$) defines a state on which geometrical operators act. The constants $\ell_j$ are only dependent on the spin $j$ and the single vertex tensor $A$. From the locality of the projectors $P_x^j$ and the product structure of the Bell pairs, we can calculate the expectation value of $L_\gamma$ on multiple layers of HITs constructing the whole holographic state $\ket{\psi}$. In this case, we get similar terms as above, but have to sum over all intersections points $x \in \Gamma \cap \gamma$ of the graph $\Gamma$ underlying the holographic state and the curve $\gamma$. However, since the tensor network only consists of Bell pairs, all holonomies cancel and the contraction reduces to the level of single tensors again. This yields
\begin{align}
    \bra{ \psi } L_\gamma \ket{ \psi } &= \sum_{x\in \Gamma\cap \gamma} \sum_j \sqrt{j(j+1)} \lvert \braket{ \psi | P_x^j | \psi }\rvert^2 \nonumber \\
    &= L_\Gamma(\gamma) \sum_j \ell_j,
    \label{eq:graph_length}
\end{align}
with $L_\Gamma(\gamma) := \abs{\Gamma\cap\gamma}$ being the graph length of $\gamma$. Since the factor $\sum_j \ell_j$ is fixed by the single HIT, Eq.~\eqref{eq:graph_length} proves the following lemma.
\begin{lemma}
    Let $\ket{\psi}$ be a holographic tensor network state built from a HIT $\ket{A}$ that is a tensor product of Bell pairs and let $\gamma$ be a curve on the Poincaré disc. The quantum geometrical length and the graph length of $\gamma$ coincide 
    \begin{align}
        \braket{L_\gamma} = c_A L_\Gamma(\gamma),
        \label{eq:length_linear}
    \end{align}
    up to a constant $c_A = \sum_j \ell_j$ that is determined by the HIT $A$.
\end{lemma}
Let us consider the 3-valent HIT from Example \ref{ex:LR}. The only relevant spin subspaces are $j \in \{0,1\}$, since the holonomies are 4-dimensional (that is, they act to a two-qubit space). Since the contribution $\ell_0$ vanishes, the proportionality constant in this case is particularly simple, $c_A = \ell_1$. The spin 1 projector is a simple symmetrization of two spin $\frac{1}{2}$ indices and has the form $P_x^1 = \frac{1}{3} \ket{\vcenter{\hbox{\spinone}}}\bra{\vcenter{\hbox{\spinone}}}$, with a factor $\frac{1}{3}$ ensuring normalization and idempotence. If we plug in this example into Eq.~\eqref{eq:length_contribution_A}, we get
\begin{align}
    \bra{ \vcenter{\hbox{\threeVAplain{0.15}}}} L_\gamma \ket{ \vcenter{\hbox{\threeVAplain{0.15}}} } &= \sqrt{2} \lvert \braket{ \vcenter{\hbox{\threeVAplain{0.15}}} | \vcenter{\hbox{\threeVAspin{0.15}}} \hspace{-1ex} } \rvert^2 = \frac{\sqrt{2}}{16} \lvert \vcenter{\hbox{\traceid{}}}\rvert^2 \nonumber \\
    &= \frac{9\sqrt{2}}{16} = \ell_1,
\end{align}
where $\sqrt{j(j+1)} = \sqrt{2}$ and we assume the state to be normalized $|\braket{\vcenter{\hbox{\threeVAplain{0.15}}}| \vcenter{\hbox{\threeVAplain{0.15}}}}|^2 = 1$. Constructing higher spin projectors $P_x^j$ is straightforward, but quickly becomes tedious. We discuss a general procedure to derive $P_x^j$ in Appendix \ref{app:higher_spin}.

Note that the above discussion is specific to non-crossing Bell pairs, which form a basis of the SU(2)-invariant space. Tensors built from a superposition of such Bell pair products are more complicated, as the combinations of Bell pair contractions scales exponential in the number of vertex tensors $A$. For instance already a superposition like $A = (\ket{ \vcenter{\hbox{\threeVAplain{0.15}}}} + \ket{ \vcenter{\hbox{\threeVAplainII{0.15}}} })/\sqrt{2}$ requires the evaluation of $2^N$ (possibly different) matrix elements when $N$ is the number of vertices on the Poincaré disc. As we have not found any indication for the existence of such HITs, we have to leave a possible deviation of the geodesic length from the graph length as an open question.

In addition to the geodesic length, we are also able to calculate the variance of the geodesic length $L_\gamma$ evaluating the expectation value $\braket{L_\gamma^2}$. The squared operator $L_\gamma^2$ decomposes into products of the local operators $\braket{L_{e} L_{e'}}$ acting on edges $e$ and $e'$ that intersect $\gamma$. This product, in turn, decomposes to the joint application of projectors onto spins $j$ and $k$, which is a simple calculation in the case of coincidence of the geodesic and the entanglement wedge $\mathcal{E}_\mathcal{A}$ corresponding to a boundary region $\mathcal{A}$. Leaving out the eigenvalues $\sqrt{j(j+1)}\sqrt{k(k+1)}$, the relevant terms in $\braket{L_\gamma^2}$, acting on $e \neq e'$, are of the form
\begin{align}
    &\Tr \left( \vcenter{\hbox{\lengthvariance}} \right) = \frac{\ell_j \ell_k}{\sqrt{j(j+1)}\sqrt{k(k+1)}}.
    \label{eq:length_correlation}
\end{align}
In the case $e = e'$, we get a contribution $\sqrt{j(j+1)} \ell_j$ using the projective property of the symmetrization. With this, the second moment of the length operator becomes
\begin{align}
    &\braket{L_\gamma^2} = L_\Gamma(\gamma) \sum_j \sqrt{j(j+1)} \ell_j + L_\Gamma(\gamma) (L_\Gamma(\gamma) - 1) c_A^2 \nonumber \\
    &= \braket{L_\gamma}^2 + L_\Gamma(\gamma) \left( \sum_j \ell_j \left( \sqrt{j(j+1)} - c_A \right) \right)
\end{align}
using $c_A = \sum_{j} \ell_j$ from Eq.~\eqref{eq:length_linear}. Again, for our example of a (7,3) tiling, we can insert $c_A = \ell_1 = \frac{9}{16} \sqrt{2}$ and $j=1$, which yields a variance ${\rm Var}(L_\gamma) = \braket{L_\gamma^2} - \braket{L_\gamma}^2 = \frac{63}{128} L_\Gamma(\gamma)$. \change{Keeping track of the Planck-scale spectrum, $\braket{L} \sim G_N \sum_j j$, one can calculate $\mathrm{Var}(L) = O(G_N \braket{L})$, which corresponds to $\mathrm{Var}(L) = O(G_N^2)$ at fixed spins. In standard holographic expectations for fluctuations around a classical geometry \cite{Faulkner_2013}, one fixes the macroscopic length, which corresponds to the semiclassical limit $G_N \to 0$ with $j \to \infty$. In this case, one finds $\mathrm{Var}(L) = O(G_N)$.}

It has been shown in Refs.~\cite{evenbly2017hyperinvariant, Steinberg_2023} that hyperinvariant tensor networks do not always fulfil the RT formula (for instance in a (5,4) tiling). There are cases where the causal wedge is slightly larger than the entanglement wedge and thus the tensors separated by the geodesic do not trace to identity. In these cases, the entanglement wedges of the parties $\mathcal{A}$ and $\mathcal{A}^\complement$ are separated by a strip between the endpoints of $\mathcal{A}$. The length operator acts on holonomies within the excluded strip as opposed to Eq.~\eqref{eq:length_correlation}, where it acts at the boundary of the entanglement wedge. As a consequence, the trace does not just yield symmetrized identities, but there are surviving tensors $A, B$ that need to be contracted manually, yielding a slightly deviating result, in general. This happens for instance in certain cuts of a (5,4) tiling. The same calculation as before can be carried out, with the only difference that the trace in Eq.~\eqref{eq:length_correlation} does not simplify to a product of Wilson loops, but rather to a tensor contraction of the following form:
\begin{align}
    \Tr \left( \vcenter{\hbox{\lengthvariancestrip}} \right)
\end{align}
with a strip $\sigma_\mathcal{A}$ defined by the region between the isometric boundaries $\gamma_\mathcal{A}$ and $\gamma_{\mathcal{A}^\complement}$ determined by a greedy algorithm starting from a curve along $\mathcal{A}$ or $\mathcal{A}^\complement$, respectively, see Fig.~\ref{fig:strip}.

\begin{figure}
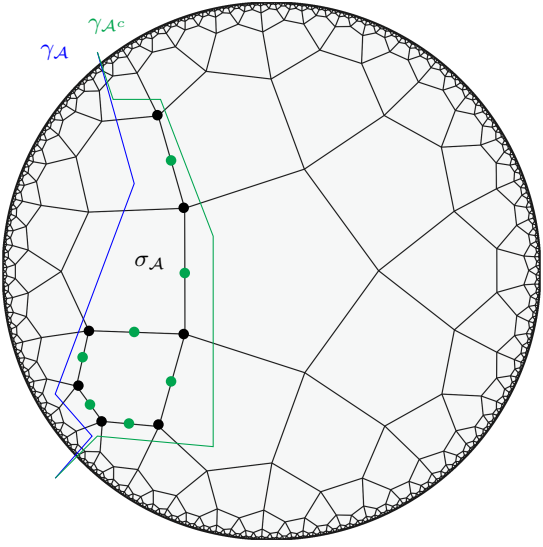

    \centering
    $\vcenter{\hbox{\SNSdiscStrip}}$
    \caption{\textbf{Equivalence defect of entanglement and causal cone.} The strip $\sigma_\mathcal{A}$ is composed of the $A$ tensors represented by the black dots and $B$ tensors in green, for instance as in Eq.~\ref{eq:LR_state4}. }
    \label{fig:strip}
\end{figure}

\subsection{Wedge Area}
The fact that every tensor has an equal area contribution is more direct and can be seen from the definition of the total area of a surface $\mathcal{S}$, $S_\mathcal{S} = \sum_{v \in \mathcal{S}} S_v$. As we consider a state that is built from copies of the same tensors $A$ and $B$, the total area of $\mathcal{S}$ is linear in the number of vertices $S_\mathcal{S} = \abs{S} S_{v_0}$ for any $v_0 \in \mathcal{S}$. It turns out that the direct calculation of $S_{v_0}$ is not instructive, so we are left with a decomposition of the HIT $A$ into spin network states. Three-valent intertwiners with spin labels $(j,k,l)$ are eigenstates of the area operator to the eigenvalue 
\begin{align}
    s^2_{jkl} := &\frac{9}{4}\left[ 2( \Delta_j \Delta_k + \Delta_j \Delta_l + \Delta_l \Delta_k) - (\Delta_j^2 + \Delta_k^2 + \Delta_l^2) \right] \nonumber \\
    &- \frac{1}{2} (\Delta_j + \Delta_k + \Delta_l)
    \label{eq:eigenarea}
\end{align}
with $\Delta_j = -j(j+1)$. Furthermore, two-valent vertices with spin $j$ and non-parallel tangent vectors are also eigenstates and contribute with an eigenvalue $j(j+1)$ as shown in Ref.~\cite{Thiemann_1998_QSD}. With the same argument as in the length calculation, the expectation value of $S_{v_0}$ can then be calculated as a local inner product using the fact that the holonomies cancel and the Bell pairs at the boundary are either contracted to loops that cancel with the normalization factor, or connect to the tensor at $v_0$. On a generic tensor $A = \sum_{j_1, ..., j_n} c_{j_1 ... j_n} A_{j_1 ... j_n}$ at $v_0$, the expectation value of the area operator reads
\begin{align}
    \braket{S_{v_0}}_A = \sum_{j_1, ..., j_n} \abs{c_{j_1 ... j_n}}^2 s_{j_1 ... j_n}
\end{align}
Let us consider the HIT for the (7,3) tiling, where $A = \vcenter{\hbox{\threeVAplain{0.15}}}$ (cf. Example \ref{ex:LR}), as an example to demonstrate the area calculation. Coupling six qubits in an $\SUT$-invariant way gives rise to a five-dimensional invariant subspace, which can be constructed by Clebsch-Gordan coupling,

\begin{align}
    \frac{1}{2}^{\otimes 6} \cong (0 \oplus 1)^{\otimes 3} \cong 0^{\oplus 5} \oplus 1^{\oplus 9} \oplus 2^{\oplus 5} \oplus 3,
\end{align}
where we chose a coupling scheme that first couples the two indices pointing into the same direction without loss of generality. In fact, such subspaces including precursors of Eqs.~(\ref{eq:73_dec}, \ref{eq:73_dec-zwei}) below have already been studied in 1932 by G. Rumer \cite{Rumer1932}. The basis states read
\begin{align}
    A_{000} &= \frac{1}{n_{000}} \, \vcenter{\hbox{\threeVAbasisI}} \quad A_{011} = \frac{1}{n_{011}} \, \vcenter{\hbox{\threeVAbasisII}} \nonumber \\[5mm]
    A_{101} &= \frac{1}{n_{101}} \,\vcenter{\hbox{\threeVAbasisIII}} \quad A_{110} = \frac{1}{n_{110}} \,\vcenter{\hbox{\threeVAbasisIV}} \nonumber \\[5mm]
    A_{111} &= \frac{1}{n_{111}} \,\vcenter{\hbox{\threeVAbasisV}}\; ,
    \label{eq:73_dec}
\end{align}
with the white square indicating symmetrization of the indices, straight lines depicting Kronecker deltas and arcs depicting the Levi-Civita symbol $i \epsilon_{ab}$ in two dimensions multiplied by an imaginary unit, such that 
\begin{align}
    \bigcirc = -2 \qquad \text{and} \qquad -\vcenter{\hbox{\binorI}} = |\;| + \vcenter{\hbox{\binorII}}.
 \label{eq:73_dec-zwei}
\end{align}
The $n_{j_1 j_2 j_3}$ in Eq.~\eqref{eq:73_dec} are normalization factors that can be calculated using the Hilbert-Schmidt inner product. A straightforward calculation yields $n_{000} = \sqrt{8}, n_{011} = n_{101} = n_{110} = \sqrt{6}$ and $n_{111} = \sqrt{3}$, as well as $\norm{A} = \sqrt{8}$. Similarly, we can calculate the overlap between the (7,3) HIT $A = \sum_{j_1 j_2 j_3} c_{j_1 j_2 j_3} A_{j_1 j_2 j_3}$ and the basis tensors $A_{j_1 j_2 j_3}$ yielding the coefficients $c_{j_1 j_2 j_3}$, for instance
\begin{align}
    c_{000} = \frac{\braket{A | A_{000}}}{\norm{A}} = \frac{1}{8} \vcenter{\hbox{\threeVAcoeffI}} = - \frac{1}{4},
    \label{eq:coeffs73}
\end{align}
Analogously, we get $c_{011} = c_{101} = c_{110} = \frac{\sqrt{3}}{4}$ and $c_{111} = - \frac{\sqrt{6}}{4}$. To calculate the area of the triangle dual to the vertex $A$, the following eigenvalues are relevant
\begin{align}
    s_{000} = 0 \qquad s_{110} = s_{101} = s_{011} = \sqrt{2} \qquad s_{111} = \sqrt{30}.
\end{align}
With this, the area contribution of a single vertex reads $\braket{S_{v_0}} = \frac{9}{16} \sqrt{2} + \frac{3}{8} \sqrt{30}$. Note that the formula for equilateral triangles is not fulfilled on the level of expectation values $\braket{S}$ and $\braket{L^2}$. In that sense, the states that we consider are not semi-classical, but truly quantum states.

\subsection{Negative Curvature}
In order to make a statement about the curvature of a region, we consider the sum of angles of an $n$-gon, which is equal to $(n-2)\pi$ in the absence of curvature. Similar to the area of a polygon (see Eq.~\eqref{eq:area_op}), we can calculate the loop quantum gravitational version of the angle between two tangent vectors $\vec X(e_1), \vec X(e_2)$ via the formula
\begin{align}
    \cos(\theta) = \frac{\braket{ \vec X(e_1)[.] | \vec X(e_2)[.] }}{\norm{\vec X(e_1)[.]} \norm{\vec X(e_2)[.]}},
\end{align}
with the tangent vectors from Eq.~\eqref{eq:grasping}. Using the Leibniz rule for tangent vectors, we can calculate the angle of the triangles defined by the 3-valent tensor $A$ from Eq.~\eqref{eq:LR_state3}. Higher valent tensors can be treated analogously. The action of $\vec X$ on $A$ reads 
\begin{align}
    \vec X(e_1)(A) = \vcenter{\hbox{\graspedA{0}}} + \vcenter{\hbox{\graspedAII{0}}} \label{eq:Xe1A}, \\
    \vec X(e_2)(A) = \vcenter{\hbox{\graspedA{-120}}} + \vcenter{\hbox{\graspedAII{-120}}},
    \label{eq:Xe2A} \\
    \vec X(e_3)(A) = \vcenter{\hbox{\graspedA{120}}} + \vcenter{\hbox{\graspedAII{120}}}.
    \label{eq:Xe3A}
\end{align}
Since the tensor is rotation invariant, the three sites as well as the three angles need to be equal. It thus suffices to do the calculation with one pair of tangent vectors. A straightforward calculation yields
\begin{align}
    \braket{\vec X(e_1)(A) | \vec X(e_2)(A)} = - \frac{1}{2} \norm{\vec X(e_1)(A)} \norm{\vec X(e_2)(A)},
    \label{eq:angle}
\end{align}
which is a consequence of the orthogonality of all but the second term from Eq.~\eqref{eq:Xe1A} and the first term from Eq.~\eqref{eq:Xe2A}, which are equal up to a minus sign. The angle contributing to the sum of angles in a triangle is related to the angle between the tangent vectors via $\alpha = \pi - \theta$, as the normal vectors can be virtually parallel-transported 
\begin{align}
    \vcenter{\hbox{\dualtriangle}}
\end{align}
Together with Eq.~\eqref{eq:angle}, this yields $\alpha = \arccos\left( \frac{1}{2} \right) = \frac{\pi}{3}$, which indicates no curvature on the level of a single vertex. However, arranging ``flat'' triangles in a (7,3) tiling creates negatively curved surfaces. For instance, the dodecagon built from one layer of a (7,3) tiling has a sum-of-angles deficit of $\pi$. This can be seen using a simple geometrical argument that identifies the angles of the dodecagon as multiples of $\alpha$, see Fig.~\ref{fig:curvature}.
\begin{figure}
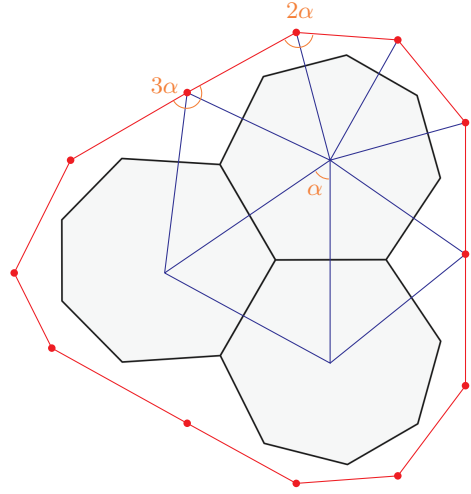

    \centering
    $\vcenter{\hbox{\anglesum}}$
    \caption{\textbf{Hyperbolic arrangement of ``flat'' triangles in a single layer.} The triangles from which the dodecagon is built and which are defined by the tensors $A$ are all identical. Although, since we draw the dodecagon in a flat projection, the triangles appear different in size.}
    \label{fig:curvature}
\end{figure}
The sum of angles of the dodecagon finally is
\begin{align}
    9\cdot(2\alpha) + 3 \cdot (3\alpha) = 9 \pi < 10 \pi = (n-2)\pi,
\end{align}
indicating negative curvature.

\subsection{Boundary Correlations}
The HIT constructions discussed above still have a free parameter $k$ that determines the number of reducible SU(2) components per quantum informational party. Viewing the boundary state as a discretization of a CFT state, there are, however, expectations towards the decay of correlations as a function of the relative angle on the boundary of the Poincaré disc. Consider two local observables $O_1$ and $O_2$, which have support on single parties on the boundary.

On products of Bell pairs as our examples in section \ref{sec:examples} suggest, the two-point correlator is easy to calculate in graphical notation. For $k=1$, the two observables can either share a Bell pair or not. This means, that they are either uncorrelated, $\braket{O_1 O_2} = \braket{O_1} \braket{O_2}$, or fully correlated, meaning
\begin{align}
    \braket{\psi | O_1 O_2 | \psi}= \vcenter{\hbox{\correlations}}.
\end{align}
This singular correlation structure is not expected to appear in a continuum limit. Rather, we would like neighboring sites to share more Bell pairs than far away sites. We can achieve this by raising $k$, the number of irreducible components per party. For $k=2$, for instance, we can have uncorrelated sites, but also partially correlated sites:
\begin{align}
    \braket{\psi | O_1 O_2 | \psi}= \vcenter{\hbox{\correlationsII}}.
\end{align}
In order to satisfy non-singular correlation structures decaying with boundary distance, we need to have larger $k$. Every party has to share a number of Bell pairs with every other party decaying in distance as $e^{-j/\xi}$, with a correlation length $\xi$ \change{for an exponential decay of correlations, or as $j^{-\alpha}$ for algebraic decay.} The number of Bell pairs per party is given by summing the partial sum of the geometric series
\begin{align}
    k = 2m \sum_{j=1}^{n/2} e^{-j/\xi} = 2m \frac{1 - e^{-n/(2\xi)}}{e^{1/\xi} - 1}
    \label{eq:partial_sum}
\end{align}
\change{assuming exponential decay of correlations,} with $m$ being the maximal number of shared Bell pairs and $n$ being the number of parties at the boundary. The factor 2 is necessary, since the sum is only capturing correlations from one point at the boundary up to the antipode. Requiring all correlations to be described by at least one shared Bell pair implies $m e^{-j/\xi} \geq 1$, which bounds $j$ and thus $m \leq e^{n/(2\xi)}$. Finally, this yields $k = 2 \frac{e^{n/(2\xi)} - 1}{e^{1/\xi} - 1}$. \change{Assuming algebraic decay of correlations, the partial sums in Eq.~\eqref{eq:partial_sum} do not admit a closed form, but a rough estimate yields $k = \mathcal{O}(n^\alpha)$. On the other hand, a system size dependent correlation length $\xi = \Omega(n)$ also implies a linear growth in fine structure, $k = \Omega(n)$.}

In order to distribute Bell pairs to different parties on the boundaries, also the valence $q$ will be lower bounded by $k$, which leaves us with a \change{necessary scaling $q = \Omega(k)$}.

\section{Conclusion}
Emergent spacetime can be loop quantum gravitational. Quantum states with the entanglement structure showing aspects of AdS/CFT duality are contained in its Hilbert space. While SU(2)-invariance prohibits tensors to have entanglement structures going beyond maximally mixed 2-marginals of one specific party with any other party, weaker isometries can be implemented using a tensor product of Bell pairs shared among potentially many parties. Those hyperinvariant tensors do not only capture holographic properties, but also yield direct access to geometric properties, such as curve length and surface area derived from LQG. We show that, on holographic states, curve lengths and the graph length frequently used in the literature coincide up to a constant factor. Similarly, the area of a surface is proven to grow linear with the number of vertices inside. These results make use of the simple structure of HITs that we found to be tensor products of shared Bell pairs. The question remains whether there are HITs that have a more sophisticated, potentially multipartite entanglement structure, and whether similar geometric relations hold for those states as well.

The nonexistence of holographic codes within the LQG framework is another insight that poses future challenges for encoding matter degrees of freedom in the bulk theory. Lifting the condition of separable 2-marginals containing the bulk index, \change{or considering quasi-crystalline tilings \cite{Boyle_2020, Boyle_2025},} provides a way around our no-go theorem and prior no-go results for covariant codes. What needs to be clarified is whether those structures allow for at least partial reconstruction of bulk information in order to call it an emergent spacetime.

Recent work also discussed the complementary recovery of type II factors in an inductive limit of perfect tensor networks \cite{chemissany2025infinitetensornetworkscomplementary} in view of the necessity of type III factors in von Neumann algebras that describe quantum field theories. We showed that enforcing SU(2)-invariance (ultimately originating from Lorentz symmetry) forbids those structures previously criticized. However, it is not clear whether a tensor network built from HITs automatically implies type III factors in an inductive limit. Viewing the LQG states discussed in this paper in an algebraic perspective remains an open question left for future work.

Finally, the AdS/CFT correspondence is a convenient but ultimately unphysical toy model for the study of emergent spacetimes. LQG provides direct relations between entanglement and geometry independently of the background manifold \cite{Bose_2017,sahlmann2025bellstatesfermionsloop, Livine_2018}. Thus, leaving the Poincaré disc but keeping isometric structures could open possibilities to similar emergent phenomena and entanglement-geometry correspondence.

\begin{acknowledgments}
The authors thank Felix Huber, Matthias Kleinmann, Andreas Leitherer, Seth Lloyd, Erickson Tjoa, Freek Witteveen, and Paolo Zanardi for helpful discussions. OG, NS, and RM thank the
Centro de Ciencias de Benasque Pedro Pascual for hosting them during the workshop ``Quantum Information'' in June/July 2025. RM and NS acknowledge support from the Austrian Science Fund (FWF) via the Cluster of Excellence ``quantA'' (Grant No.~\href{https://doi.org/10.55776/COE1}{10.55776/COE1}), the Special Research Programme ``BeyondC: Quantum Information Systems Beyond Classical Capabilities'' (Grant No.~\href{https://doi.org/10.55776/F71}{10.55776/F71}), and the Standalone Project ``Entanglement Order Parameters'' 
(Grant No.~\href{https://doi.org/10.55776/P36305}{10.55776/P36305}), by the European Union -- NextGenerationEU, and by the European Union’s Horizon 2020 research and innovation programme through Grant No.~863476 (ERC-CoG SEQUAM). FO and OG acknowledge support from the Deutsche Forschungsgemeinschaft (DFG, German Research Foundation, project number 563437167), the 
Sino-German Center for Research Promotion (Project M-0294), and the
German Federal Ministry of Research, Technology and Space (Project QuKuK, Grant No.~16KIS1618K and Project BeRyQC, Grant No.~13N17292). HS acknowledges the contribution from COST Action BridgeQG (CA23130), supported by COST (European Cooperation in Science and Technology).

The visualizations of $(p,q)$ tilings of the Poincaré disc shown in this paper have been generated using the Python package \textsc{hypertiling} \cite{hypertiling}. 
\end{acknowledgments}

\bibliography{literature}

\newpage

\appendix

\onecolumngrid
\section{Grasping Action of the 2-Area Operator}
\label{app:grasp}
The action of $S_v$ as defined in Eq.~\eqref{eq:area_op} can be calculated graphically using the grasping operation that is the action of $\vec X(e)$. The operator $\vec X(e_1) \times \vec X(e_2)$ acts by connecting the edges $e_1, e_2$ by coupling a ``virtual'' holonomy in spin 1 representation with a free spin 1 index in the middle:
\begin{align}
    \sum_{i_1, i_2} \epsilon_{i_1 i_2 i_3} X^{i_1}(e_1)X^{i_2}(e_2)\left( \vcenter{\hbox{\threeVal}} \right) = \vcenter{\hbox{\graspingEps}}
    \label{eq:graspingEps}
\end{align}
On a direct sum of irreducible representations on $e_1$ and $e_2$, this contraction acts on each factor according to a Leibniz rule. In contrast to Eq.~\eqref{eq:graspingEps}, the two grasped edges $e_1$ and $e_2$ could also coincide. Having summed over all combinations of pairs, the operator from Eq.~\eqref{eq:graspingEps} has to be applied again yielding, for instance, the contraction.
\begin{align}
    \sum_{i_1, i_2, i_3, i_4, i_5} \epsilon_{i_1 i_2 i_3} X^{i_1}(e_1)X^{i_2}(e_2) \epsilon_{i_4 i_5 i_3} X^{i_4}(e_1)X^{i_5}(e_3) \left( \vcenter{\hbox{\threeVal}} \right) = \vcenter{\hbox{\graspingEpsII}}
\end{align}

\section{Maximal Number of Maximally Mixed Marginals}\label{app:n_marginals}
From the main text, we know that $\SUT$-invariance and the AME property exclude each other for more than three parties. Nevertheless, there are planar maximally entangled states \cite{Doroudiani2020}, where all marginals containing neighboring parties are maximally mixed. In this section, we investigate the existence of such states having a global $\SUT$ symmetry and derive the maximum number of balanced bipartitions of an $\SUT$-invariant state such that the reduced states are maximally mixed.
\begin{lemma}\label{lem:bip_sep}
	Let $\ket\psi\in\HH_d^{\otimes 2n}$ be a $2n$-partite $\SUT$-invariant state and its marginals maximally mixed for three balanced bipartitions, i.e.\ the cardinality of all parts is $n$. Then, for any party $j$ there is a party $k$ such that $j$ and $k$ belong to different parts in all three of 
    the bipartitions.
\end{lemma}
\begin{proof}
    We prove the statement by contradiction. W.l.o.g.\ let $j\in\mathcal P^i_1$ for all $i\in[3]$ and let the marginals of $\ket\psi$ be maximally mixed for balanced bipartitions $\{\mathcal{P}^i_{1/2}\}_{i=1}^3$ with $\mathcal P_1^1\cup \mathcal P_1^2\cup\mathcal P_1^3=[2n]$. This implies that all two-partite reduced states $\rho^{\{j,l\}}$ with $j\neq l\in[2n]$ are maximally mixed. This is in contradiction to the assumption of $\SUT$-invariance as shown in \cref{lem:inv_mixed_nogo}.
\end{proof}
Note that three bipartitions are needed in the statement of the Lemma, because $\mathcal P_1^1\cup \mathcal P_1^2\cup\mathcal P_1^3=[2n]$ and $j\in\mathcal P^i_1$ for all $i\in[3]$ could not be fulfilled otherwise.

\begin{figure}
    \centering
    \includegraphics[width=0.3\linewidth]{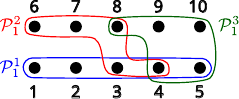}
    \caption{Party $4$ is contained in the parts $\mathcal P^1_1,\mathcal P^2_1,\mathcal P^3_1$ spanning the whole set of parties $\mathcal P^1_1\cup \mathcal P^2_1 \cup \mathcal P_1^3=\{1,\dots,10\}$. Maximally mixed marginals for these bipartitions and $\SUT$-invariance of the state exclude each other.}
    \label{fig:bip}
\end{figure}

This Lemma serves as the basis for the following Proposition.

\begin{prop}
	An $\SUT$-invariant $2n$-partite state $\ket\psi\in\HH_d^{\otimes 2n}$ has at most $2^{n-1}$ balanced bipartitions with maximally mixed marginals. This upper bound is saturated for states consisting of $n$ $d$-dimensional Bell pairs, 
    $\ket{\psi^{(d)}}$:
    \begin{equation}
    	\ket\psi=
        \ket{\psi^{(d)}}^{\otimes n}.
    \end{equation}
\end{prop}
\begin{proof}
    If the marginals of $\ket\psi$ are maximally mixed for $m>2$ balanced bipartitions $\{\mathcal{P}_{1/2}^i\}_{i=1}^m$, for each party $j\in\mathcal P^i_{1}$  there must be a party $j^\prime\in \mathcal P^i_{2}$ for all $i\in[m]$ and the other way around. This gives rise to a graph $G$ with the $m$ parties as nodes. Two parties are connected if they are chosen to not appear in the same part of every bipartition. The maximal number of balanced bipartitions with maximally mixed marginals is upper-bounded by the maximal number of balanced bipartitions such that there are no connections inside the single parts, because the conditions presented in \cref{lem:bip_sep} are necessary, but not proven to be sufficient. 
    
    If two nodes of degree greater than one are connected, this connection can be deleted since every party is still connected to at least one other party. Thus, we just need to consider graphs with star subgraphs, i.e.\ subgraphs containing a center node connected to all other nodes (\emph{leaves}) without any further edges. Let $p$ be the number of star subgraphs. Then, at most $2^{p-1}$ balanced bipartitions without connections inside the single parts as the center of every star subgraph has to be in the opposite part to the leaves. This is maximized for $p=n$ and leads to an upper bound of $2^{n-1}$ for the number of balanced bipartitions with maximally mixed marginals.

    This bound can be saturated by a state consisting of $n$ Bell pairs being distributed among the $2n$ parties. If the Bell pairs are distributed between opposite parties, the resulting state is planar maximally mixed and can be used in models for the AdS/CFT-correspondence as in Ref.~\cite{Pastawski_2015}.
\end{proof}

\section{Geometric measure of entanglement}\label{app:GM}
The geometric measure of entanglement \cite{wei_2003_geometric,weinbrenner_2025_quantifying} of a state $\ket\psi\in\otimes_{i=1}^n \HH_{d_i}$ is defined as $E_G(\ket\psi)=1-\Lambda^2(\ket\psi)$ with
\begin{equation}
    \Lambda^2(\ket\psi)=\sup_{\ket\phi=\ket{\phi_1}\otimes\dots\otimes\ket{\phi_n}\in\otimes_{i=1}^n \HH_{d_i}} \abs{\braket{\phi|\psi}}^2
\end{equation}
being the maximum overlap with separable states. \change{We use the geometric measure to quantify the entanglement in the multipartite states presented in the examples for HITs in \cref{sec:examples} by exploiting the decomposition into (bipartite) Bell pairs.}
\begin{lemma}
    Consider a HIT network with $n$ sites where each site can share $d$-dimensional bipartite singlet states
    \begin{equation}
        \ket{\psi^-}=\sum_{i=0}^{d-1}(-1)^i\ket{i}\ket{d-i-1}
    \end{equation}
    with each of the other sites. The resulting global state $\ket\psi$ consists of in total $m$ distributed singlet states. The geometric measure of entanglement of the state $\ket\psi$ is given by
\begin{equation}
    E_G(\ket\psi)=1-\Lambda^2(\ket\psi)=1-d^{-m}.
\end{equation}
\end{lemma}
\begin{proof}
    First, one can separate one singlet state connecting sites $i,j$ from the rest of state: $\ket{\psi}=\ket{\psi^-}_{i,j}\otimes\ket{\Tilde{\psi}}$. As any bipartite pure state, the singlet state can be written in a suitable basis as a state having non-negative coefficients only. For such so-called non-negative states multiplicativity of the maximum overlap with the separable states is known \cite{zhu2010additivity}:
    \begin{equation}
        \Lambda^2(\ket\psi)=\Lambda^2(\ket{\psi^-})\Lambda^2(\ket{\Tilde{\psi}})= d^{-1}\Lambda^2(\ket{\Tilde{\psi}}).
    \end{equation}
    Iterating this argument leads to $\Lambda^2(\ket{\psi})=d^{-m}$ implying the claim.
\end{proof}
\section{Connection to the proof in Hayden et al. (2021)}\label{app:err_corr_proof}

In this appendix, we connect our proof of \cref{lem:inv_mixed_nogo} to a similar proof in the context of quantum error correction presented in Ref~\cite{Hayden_2021}. The latter is formulated on the level of channels and can be translated to our problem via the Choi-Jamio{\l}kowski isomorphism.  

Suppose there is a state $\rho=\pro\psi$ with $\ket\psi\in\otimes_{i=1}^n \HH_{d_i}$ and marginals $\rho_{\{k,l\}^\complement}=\frac{\ID}{d_k}\otimes\rho_l$ for all $l\neq k$. In addition, it is symmetric under the action of a continuous symmetry group $G$ with representation $R_1(g)\otimes\ldots\otimes R_n(g)$ and one generator of $R_k(g)$ being non-trivial. W.l.o.g.\ we choose $k=1$.\\
Since $\tr_{\{1\}^\complement}\pro{\psi}=\frac{\ID}{\dim(\HH_1)}$, we can define a channel $\mathcal{E_\psi}: \mathcal{B}(\HH_1)\to\mathcal{B}(\HH_2\otimes\dots\otimes\HH_n)$ via the Choi-Jamio{\l}kowski isomorphism such that
\begin{equation}
    \pro{\psi}=\sum_{i,j}\kb i j \otimes \EE_\psi(\kb ij).
\end{equation}
The channel $\EE$ is covariant under the action of the group $G$ with respect to the representation $\bar U_1$ on the input and $U_2\otimes\ldots\otimes U_n$ on the output (e.g.\ \cite[Lemma~11]{gschwendtner2021programmability}).
The property $\tr_{\{1,k\}^\complement}\pro{\psi}=\frac{\ID}{\dim(\HH_1)}\otimes\rho_k$ for all $k\in\{2,\dots,n\}$ implies that a single subsystem of the output is not correlated with the input and $\EE_\psi$ is an encoding of a quantum error correcting code. However, this is excluded for covariant encodings \cite{Hayden_2021}. Thus, the state $\ket\psi$ with the mentioned properties cannot exist. Setting $G=\UO$ leads to a proof of \cref{lem:inv_mixed_nogo}.

\section{Projectors for Higher Spin Subspaces}
\label{app:higher_spin}
In this section, we sketch the derivation of the projectors $P_x^j$ onto spin $j$ subspaces used in the calculation of geodesic lengths in Section \ref{sec:geodesic}. In the simplest non-trivial case, every holonomy acts on a tensor product space of $k=2$ qubits with spin coupling $\frac{1}{2} \otimes \frac{1}{2} \cong 0 \oplus 1$ and the projectors for singlet $P^0 = \frac{1}{4} \ket{ \vcenter{\hbox{\Pzero}}} \bra{\vcenter{\hbox{\Pzero}}}$ and triplet subspace $P^1 = \frac{1}{3} \ket{\vcenter{\hbox{\Pone}}} \bra{\vcenter{\hbox{\Pone}}}$. We will omit the subscript $x$ that indicates the point at which the projector acts, as it remains unchanged within the following construction and is not essential. The normalization factors are calculated requiring $(P^j)^2 = P^j$, which yields
\begin{align}
    (P^0)^2 &= \frac{1}{16} \ket{ \vcenter{\hbox{\Pzero}}} \braket{\vcenter{\hbox{\Pzero}} | \vcenter{\hbox{\Pzero}}} \bra{\vcenter{\hbox{\Pzero}}} = \frac{1}{4} \ket{ \vcenter{\hbox{\Pzero}}} \bra{\vcenter{\hbox{\Pzero}}} = P^0, \\
    (P^1)^2 &= \frac{1}{9} \ket{\vcenter{\hbox{\Pone}}} \braket{\vcenter{\hbox{\Pone}} | \vcenter{\hbox{\Pone}}} \bra{\vcenter{\hbox{\Pone}}} = \frac{1}{3} \ket{\vcenter{\hbox{\Pone}}} \bra{\vcenter{\hbox{\Pone}}} = P^1,
\end{align}
using the Hilbert-Schmidt norm. Higher spin projectors can be built iteratively by tensoring another spin $\frac{1}{2}$. The corresponding projectors are constructed from the projectors corresponding to the $k-1$ qubit case by coupling the $k^\text{th}$ spin either into a singlet, which reduces the total spin by $\frac{1}{2}$, or adding a spin $\frac{1}{2}$ with a symmetrization to the existing spin $j$. 

For instance the $k=3$ qubit Hilbert space decomposes into $(0 \oplus 1) \otimes \frac{1}{2} \cong \frac{1}{2} \oplus \frac{1}{2} \oplus \frac{3}{2}$. Evidently, the spin $\frac{1}{2}$ subspace is degenerate. The eigenstates of the corresponding projectors are
\begin{align}
    \vcenter{\hbox{\PonehalfI}} \text{ and } \vcenter{\hbox{\PonehalfII}} \text{ for } j = \frac{1}{2} \\
    \vcenter{\hbox{\Pthreehalf}} = \vcenter{\hbox{\PthreehalfII}} \text{ for } j = \frac{3}{2}.
\end{align}
By construction, the eigenstates are orthogonal. As before, the normalization factor is given by the squared Hilbert-Schmidt norm of the eigenstates, which yields
\begin{align}
    P^{\frac{1}{2}} &= \frac{1}{8} \ket{\vcenter{\hbox{\PonehalfI}}} \bra{\vcenter{\hbox{\PonehalfI}}} + \frac{1}{6} \ket{\vcenter{\hbox{\PonehalfII}}} \bra{\vcenter{\hbox{\PonehalfII}}} \label{eq:Ponehalf} \\
    P^{\frac{3}{2}} &= \frac{1}{4} \ket{\vcenter{\hbox{\PthreehalfII}}} \bra{\vcenter{\hbox{\PthreehalfII}}}.
\end{align}
The induction step $k \mapsto k+1$ is thus two-fold. First, couple the $(k+1)^\text{st}$ spin $\frac{1}{2}$ to the rest of the (already reduced) $k$ spins. Second, for every coupling $j \mapsto j - \frac{1}{2}$, contract a singlet state tensor $\vcenter{\hbox{\Pzero}}$ to the symmetrized spin $j$ subspace, and for every coupling $j \mapsto j + \frac{1}{2}$ tensor multiply the $(k+1)^\text{st}$ spin and symmetrize it together with the $2j$ symmetrized indices using the symmetrizer $\vcenter{\hbox{\spinj}}$.

It is easy to see that the maximal spin contribution, $j=\frac{k}{2}$, will always be one-dimensional and of the form $P^{k/2} = \frac{1}{2j+1} \ket{\vcenter{\hbox{\spinj}}} \bra{\vcenter{\hbox{\spinj}}}$. Lower spins $j$, however, can be higher-dimensional (see Eq.~\eqref{eq:Ponehalf}) and, although contributing with the same eigenvalue of the length operator, can have different overlaps with the HIT in question.

\end{document}